\documentclass[11pt,letterpaper,reqno]{amsart}

\title[Symplectic Semiclassical Wave Packet Dynamics II]{Symplectic Semiclassical Wave Packet Dynamics II:\\Non-Gaussian States}
\author{Tomoki Ohsawa}
\address{Department of Mathematical Sciences, The University of Texas at Dallas, 800 W Campbell Rd, Richardson, TX 75080-3021}
\email{tomoki@utdallas.edu}
\date{\today}

\keywords{Semiclassical mechanics, Semiclassical wave packets, Hamiltonian dynamics, Symplectic geometry}

\subjclass[2010]{37J15, 37J35, 70G45, 70H06, 70H33, 81Q05, 81Q20, 81Q70, 81S10}

\usepackage{graphicx,mathrsfs,paralist,subfigure,caption}
\usepackage[margin=1in, marginpar=.5in]{geometry}

\usepackage{tikz-cd}

\usepackage[numbers,sort&compress]{natbib}

\usepackage[colorlinks=false]{hyperref}

\theoremstyle{plain}
\newtheorem{theorem}{Theorem}[section]

\newtheorem{proposition}[theorem]{Proposition}

\theoremstyle{definition}

\theoremstyle{remark}
\newtheorem{remark}[theorem]{Remark}

\def\od#1#2{\dfrac{{\rm d}#1}{{\rm d}#2}}
\def\pd#1#2{\dfrac{\partial #1}{\partial #2}}
\def\tpd#1#2{\partial #1/\partial #2}

\def\parentheses#1{{\left(#1\right)}}
\def\brackets#1{{\left[#1\right]}}
\def\braces#1{{\left\{#1\right\}}}

\def\tr{\mathop{\mathrm{tr}}\nolimits}

\def\norm#1{{\left\|#1\right\|}}
\def\abs#1{{\left|#1\right|}}
\def\DS{\displaystyle}
\def\R{\mathbb{R}}
\def\C{\mathbb{C}}
\def\N{\mathbb{N}}
\def\defeq{\mathrel{\mathop:}=}

\def\setdef#1#2{{\left\{ #1 \ |\ #2 \right\}}}
\def\ip#1#2{{\left\langle#1,#2\right\rangle}}
\def\tip#1#2{{\langle#1,#2\rangle}}
\def\bigip#1#2{{\bigl\langle#1,#2\bigr\rangle}}
\def\exval#1{{\left\langle#1\right\rangle}}
\def\texval#1{\langle#1\rangle}
\def\diag{\operatorname{diag}}

\renewcommand{\Re}{\operatorname{Re}}
\renewcommand{\Im}{\operatorname{Im}}
\def\eps{\varepsilon}

\def\Sp{\mathsf{Sp}}

\def\U{\mathsf{U}}

\def\Mat{\mathsf{M}}

\def\d{\mathbf{d}}
\def\ins#1{{\bf i}_{#1}}

\def\rmi{{\rm i}}

\begin{document}

\footskip=.6in

\begin{abstract}
  We generalize our earlier work on the symplectic/Hamiltonian formulation of the dynamics of the Gaussian wave packet to non-Gaussian semiclassical wave packets.
  We find the symplectic forms and asymptotic expansions of the Hamiltonians associated with these semiclassical wave packets, and obtain Hamiltonian systems governing their dynamics.
  Numerical experiments demonstrate that the dynamics give a very good approximation to the short-time dynamics of the expectation values computed by a method based on Egorov's Theorem or the Initial Value Representation.
\end{abstract}

\maketitle

\section{Introduction}
\subsection{Dynamics of Gaussian and Semiclassical Wave Packets}
The Gaussian wave function is one of the most ubiquitous wave functions in quantum mechanics.
The most familiar form of Gaussian wave function appears as the ground state of the harmonic oscillator.
Gaussians also appear in many forms and play significant roles in quantum dynamics or time-dependent quantum mechanics as well; see, e.g., \citet{Ta2007}.
One of the most significant results regarding Gaussians in quantum dynamics is that the Gaussian wave packet
\begin{equation}
  \label{eq:chi_0}
  \chi_{0}(q,p,\mathcal{A},\mathcal{B},\phi,\delta; x)
  = \exp\braces{ \frac{\rmi}{\hbar}\brackets{ \frac{1}{2}(x - q)^{T}(\mathcal{A} + \rmi\mathcal{B})(x - q) + p \cdot (x - q) + (\phi + \rmi \delta) } }
\end{equation}
gives an exact solution of the Schr\"odinger equation if the potential is quadratic and the parameters $(q,p,\mathcal{A},\mathcal{B},\phi,\delta)$ satisfy the following set of ordinary differential equations~(see~\citet{He1975a,He1976b,He1981}, \citet{Ha1980,Ha1998}, and \citet{Li1986}):
\begin{equation}
  \label{eq:Heller}
  \begin{array}{c}
    \DS
    \dot{q} = \frac{p}{m},
    \qquad
    \dot{p} = -D_{q}V(q),
    \medskip\\
    \DS
    \dot{\mathcal{A}} = -\frac{1}{m}(\mathcal{A}^{2} - \mathcal{B}^{2}) - D^{2}V(q),
    \qquad
    \dot{\mathcal{B}} = -\frac{1}{m}(\mathcal{A}\mathcal{B} + \mathcal{B}\mathcal{A}),
    \medskip\\
    \DS
    \dot{\phi} = \frac{p^{2}}{2m} - V(q) - \frac{\hbar}{2m} \tr\mathcal{B},
    \qquad
    \dot{\delta} = \frac{\hbar}{2m} \tr\mathcal{A}.
  \end{array}
\end{equation}

The parameters $(q,p) \in T^{*}\R^{d}$ may be thought of as the position and momentum in the classical sense:
In fact the first two equations are nothing but the classical Hamiltonian system and is decoupled from the rest; they also give the expectation values of the position and momentum operators with respect to the Gaussian~\eqref{eq:chi_0} if it is normalized, i.e., if $\norm{\chi_{0}}=1$ then $\ip{ \chi_{0} }{ \hat{x} \chi_{0} } = q$ and $\ip{ \chi_{0} }{ \hat{p} \chi_{0} } = p$, where $\ip{\,\cdot\,}{\,\cdot\,}$ is the standard inner product on $L^{2}(\R^{d})$, $\hat{x}$ is position operator, i.e., the multiplication by the position vector $x$, and $\hat{p} = -\rmi\hbar\tpd{}{x}$ is the momentum operator.
The matrices $(\mathcal{A},\mathcal{B})$ quantify the uncertainties in the position and momentum, and live in the so-called Siegel upper half space~\cite{Si1943}
\begin{equation}
  \label{eq:Sigma_d}
  \Sigma_{d} \defeq 
  \setdef{ \mathcal{A} + {\rm i}\mathcal{B} \in \Mat_{d}(\mathbb{C}) }{ \mathcal{A}, \mathcal{B} \in \Mat_{d}(\R),\, \mathcal{A}^{T} = \mathcal{A},\, \mathcal{B}^{T} = \mathcal{B},\, \mathcal{B} > 0 },
\end{equation}
i.e., the set of symmetric (in the real sense) $d \times d$ complex matrices with positive-definite imaginary parts; this guarantees that $\chi_{0}$ is an element in $L^{2}(\R^{d})$.
The parameter $\phi \in \mathbb{S}^{1}$ is the phase factor and $\delta \in \R$ controls the norm of $\chi_{0}$ as the square of the norm of $\chi_{0}$ is given by
\begin{equation}
  \label{eq:N}
  \mathcal{N}_{\hbar}(\mathcal{B},\delta) \defeq \norm{\chi_{0}}^{2} = \sqrt{ \frac{(\pi\hbar)^{d}}{\det \mathcal{B}} }\, \exp\parentheses{ -\frac{2\delta}{\hbar} }.
\end{equation}

\citet{Ha1980,Ha1998} came up with an orthonormal basis $\{\chi_{n}\}_{n\in\N_{0}^{d}}$ for $L^{2}(\R^{d})$ whose ground state with $n = 0$ is the normalized version of the Gaussian~\eqref{eq:chi_0}; see Section~\ref{ssec:Hagedorn} below for a brief summary of its construction by ladder operators.
It was also shown that each $\chi_{n}$ gives an exact solution of the Schr\"odinger equation with quadratic potential if the parameters evolve according to \eqref{eq:Heller}.
Moreover, even with {\em non-quadratic} potentials, \citeauthor{Ha1998} gave, under some technical assumptions, an asymptotic error estimate as $\hbar \to 0$ of the approximations by certain linear combinations of the basis elements---each of which is evolving in time according to \eqref{eq:Heller}---to the solution of the Schr\"odinger equation.

\subsection{Previous Work and Motivation}
In our previous work~\cite{OhLe2013}, we followed \citet{FaLu2006} and \citet{Lu2008} to come up with the symplectic-geometric/Hamiltonian formulation of the dynamics of the Gaussian wave packet~\eqref{eq:chi_0} as follows:
The parameters $(q,p,\mathcal{A},\mathcal{B},\phi,\delta)$ associated with the Gaussian~\eqref{eq:chi_0} live in the manifold $\mathcal{M} \defeq T^{*}\R^{d} \times \Sigma_{d} \times \mathbb{S}^{1} \times \R$.
But then we can induce a natural symplectic structure $\Omega_{\mathcal{M}}^{(0)}$ and a Hamiltonian function $H^{(0)}$ on $\mathcal{M}$ by exploiting the Hamiltonian formulation of the Schr\"odinger equation (see Section~\ref{ssec:Hamiltonian_Schroedinger} below).
This results in the Hamiltonian system $\ins{X_{H^{(0)}}} \Omega_{\mathcal{M}}^{(0)} = \d{H^{(0)}}$, and gives almost the same set of equations as \eqref{eq:Heller}---the only difference being that the second equation is replaced by
\begin{equation}
  \label{eq:pdot-chi_0}
  \dot{p} = -D_{q}V^{(0)}_{\hbar}(q,\mathcal{B}),
\end{equation}
where the potential $V^{(0)}_{\hbar}$ has an $O(\hbar)$ correction term to the classical one:
\begin{equation*}
  V^{(0)}_{\hbar}(q,\mathcal{B}) \defeq
  V(q) + \frac{\hbar}{4}\tr\parentheses{ \mathcal{B}^{-1} D^{2}V(q) },
\end{equation*}
and hence the dynamics of $(q,p)$ does not satisfy the classical Hamiltonian system any more.
Numerical experiments suggest that the $(q,p)$ dynamics of our system gives a better approximation than \eqref{eq:Heller} to the short time dynamics of the expectation values
\begin{equation*}
 \exval{\hat{x}}(t) \defeq \ip{ \psi(t,\,\cdot\,) }{ \hat{x} \psi(t,\,\cdot\,) }
 \quad\text{and}\quad
 \exval{\hat{p}}(t) \defeq \ip{ \psi(t,\,\cdot\,) }{ \hat{p} \psi(t,\,\cdot\,) }
\end{equation*}
of the position and momentum with respect to the solution $\psi(t,x)$ to the initial value problem of the Schr\"odinger equation
\begin{equation}
  \label{eq:SchroedingerEq-coordinates}
  \rmi\hbar\pd{}{t}\psi(t,x) = \hat{H}\psi(t,x)
\end{equation}
with the Gaussian initial condition
\begin{equation*}
  \psi(0,x) = \chi_{0}(q(0),p(0),\mathcal{A}(0),\mathcal{B}(0),\phi(0),\delta(0); x),
\end{equation*}
where $\hat{H}$ is the standard Schr\"odinger operator
\begin{equation}
  \label{eq:SchroedingerOp}
  \hat{H} \defeq -\frac{\hbar^{2}}{2m} \Delta + V(x).
\end{equation}

The main motivation for this work is to extend our approach to non-Gaussian elements of the semiclassical wave packets $\{ \chi_{n} \}_{n\in\N_{0}^{d}}$, i.e., we would like to generalize our work~\cite{OhLe2013} done for $n = 0$ (i.e., the Gaussian~\eqref{eq:chi_0}) to those elements with $n \neq 0$.
Since the semiclassical wave packets $\{ \chi_{n} \}_{n\in\N_{0}^{d}}$ provide an orthogonal basis for $L^{2}(\R^{d})$, our extension opens the door to new semiclassical approximation methods for the Schr\"odinger equation, potentially offering improvements on the results obtained by \citet{Ha1980,Ha1998}.

The main difficulty in extending our approach to the non-Gaussian elements $\chi_{n}$ ($n \neq 0$) is that there is no known explicit formula for $\chi_{n}$ that is valid for any $n$.
The difficulty is exacerbated in the multi-dimensional case, i.e., $d > 1$:
Unlike the Hermite functions, the multi-dimensional semiclassical wave packets {\em cannot} be written as products of the one-dimensional components.
In other words, the only practical way to come up with an explicit expression for $\chi_{n}$ for a given $n \in \N_{0}^{d}$ is to apply the associated raising operator $|n|$ times to the Gaussian $\chi_{0}$ for the given dimension $d$.
This makes those calculations involving $\chi_{n}$ for an arbitrary multi-index $n \in \N_{0}^{d}$ particularly cumbersome.
The calculations of the symplectic form $\Omega_{\mathcal{M}}^{(0)}$ and Hamiltonian $H^{(0)}$ performed in our previous work~\cite{OhLe2013} were fairly straightforward because $\chi_{0}$ is a Gaussian.
However, mimicking the same calculations for an arbitrary $n\in\N_{0}^{d}$ is not feasible because of the above difficulty in obtaining an explicit expression for $\chi_{n}$ with an arbitrary $n \in \N_{0}^{d}$.

\subsection{Main Results}
We focus on those semiclassical wave packets $\{\chi_{n}\}_{n\in\N_{0}^{d}}$ that are parametrized by the same parameters $(q,p,\mathcal{A},\mathcal{B},\phi,\delta)$ as the Gaussian $\chi_{0}$, and circumvent the above difficulty by proving those recurrence relations that hold between the symplectic forms and Hamiltonians associated with $\{\chi_{n}\}_{n\in\N_{0}^{d}}$. 
Then the symplectic form $\Omega_{\mathcal{M}}^{(n)}$ and Hamiltonian $H^{(n)}$ for an arbitrary $n\in\N_{0}^{d}$ follow by induction; see Propositions~\ref{prop:Omega_M} and \ref{prop:H_M}.
As a result, we can formulate the Hamiltonian system associated with the semiclassical wave packet $\chi_{n}$ for an arbitrary $n \in \N_{0}^{d}$; see Theorem~\ref{thm:Hamiltonian_system}.

We also extend our previous results~\cite{OhLe2013} on the symplectic reduction of the dynamics of the Gaussian wave packet $\chi_{0}$ to the dynamics of an arbitrary semiclassical wave packets $\chi_{n}$ with $n \in \N_{0}^{d}$.
This results in a Hamiltonian system on the reduced symplectic manifold $T^{*}\R^{d} \times \Sigma_{d}$; see Theorem~\ref{thm:ReducedSemiclassicalSystem}.
The reduced symplectic structure takes a much simpler and suggestive form that carries an $O(\hbar)$ correction term to the classical one, and the same goes with the reduced Hamiltonian; that is, it reveals the quantum correction as an $O(\hbar)$ perturbation to the classical Hamiltonian system.

Numerical experiments with a simple one-dimensional test case demonstrate that the these Hamiltonian systems provide very good approximations to the short-time dynamics of those expectation values $\exval{\hat{x}}(t)$ and $\exval{\hat{p}}(t)$ computed by Egorov's Theorem~\cite{Eg1969,CoRo2012,KeLaOh2016} or the Initial Value Representation (IVR) method~\cite{Mi1970,Mi1974b,WaSuMi1998,Mi2001} with $\chi_{n}$ as the initial wave functions for several $n$.
The IVR is a popular method for computing such expectation values and is shown to have $O(\hbar^{2})$ asymptotic accuracy by Egorov's Theorem.

\subsection{Outline}
We start with a brief review of the semiclassical wave packets of \citet{Ha1980,Ha1998} in Section~\ref{sec:Hagedorn_wave_packets}.
We present two different parametrizations of the wave packets:
One is that used by \citeauthor{Ha1998} and the other with the same set of parameters as \eqref{eq:chi_0}; we use the latter throughout the paper as in our earlier work~\cite{OhLe2013}.
In Section~\ref{sec:Symplectic_Structures}, we find the symplectic forms $\{ \Omega_{\mathcal{M}}^{(n)} \}_{n\in\N_{0}^{d}}$ associated with the semiclassical wave packets $\{ \chi_{n} \}_{n\in\N_{0}^{d}}$.
In Section~\ref{sec:Hamiltonian_Dynamics_of_SWP}, we find the semiclassical Hamiltonians $\{ H^{(n)} \}_{n\in\N_{0}^{d}}$ and the Hamiltonian systems associated with the semiclassical wave packets $\{ \chi_{n} \}_{n\in\N_{0}^{d}}$.
In Section~\ref{sec:Symplectic_Reduction}, we perform the symplectic reduction mentioned above to simplify the formulations.
Finally, in Section~\ref{sec:Numerical_Results}, we show numerical results of a simple test case comparing our solutions with the classical solution and those obtained by an Egorov/IVR-type method.

\section{The Semiclassical Wave Packets}
\label{sec:Hagedorn_wave_packets}
\subsection{Two Parametrizations and the Siegel Upper Half Space}
\citet{Ha1980, Ha1981, Hagedorn1985, Ha1998} uses a different parametrization for the elements $\mathcal{C} = \mathcal{A} + \rmi\mathcal{B}$ in the Siegel upper half space $\Sigma_{d}$ defined in \eqref{eq:Sigma_d}.
Namely the matrix $\mathcal{C}$ in the Gaussian wave packet~\eqref{eq:chi_0} is replaced by $P Q^{-1}$ to have
\begin{equation}
  \label{eq:chi_0-Hagedorn}
  \chi_{0}(q,p,Q,P,\phi,\delta; x) = \exp\braces{ \frac{\rmi}{\hbar}\brackets{ \frac{1}{2}(x - q)^{T}P Q^{-1}(x - q) + p \cdot (x - q) + (\phi + \rmi\delta) } },
\end{equation}
where $Q, P \in \Mat_{d}(\C)$, i.e., $d \times d$ complex matrices, that satisfy
\begin{equation}
  \label{eq:Q_P-Hagedorn}
  Q^{T}P - P^{T}Q = 0,
  \qquad
  Q^{*}P - P^{*}Q = 2\rmi I_{d},
\end{equation}
where $I_{d}$ is the $d \times d$ identity matrix.
It is pointed out by \citet[Section~V.1]{Lu2008} that this is a parametrization of elements in the symplectic group $\Sp(2d,\R)$ in the following way:
\begin{equation*}
  \Sp(2d,\R)
  = \setdef{
    \begin{bmatrix}
      \Re Q & \Im Q \\
      \Re P & \Im P
    \end{bmatrix} \in \Mat_{2d}(\R)
  }
  {
    \begin{array}{c}
      Q, P \in \Mat_{d}(\C),\ Q^{T}P - P^{T}Q = 0, \smallskip\\
      Q^{*}P - P^{*}Q = 2\rmi I_{d}
    \end{array}
  }.
\end{equation*}
In fact, one can show that if $(Q,P)$ satisfies \eqref{eq:Q_P-Hagedorn} then $Q$ is invertible and also $P Q^{-1} \in \Sigma_{d}$; see, e.g., \cite[Lemma~V.1.1 on p.~124]{Lu2008}.
However, for a given $\mathcal{A} + \rmi\mathcal{B} \in \Sigma_{d}$, the corresponding $(Q, P)$ satisfying \eqref{eq:Q_P-Hagedorn} and $P Q^{-1} = \mathcal{A} + \rmi\mathcal{B}$ is not unique:
For example, one finds that, by setting
\begin{equation*}
   \begin{bmatrix}
    \Re Q_{0} & \Im Q_{0} \\
    \Re P_{0} & \Im P_{0}
  \end{bmatrix}
  =
  \begin{bmatrix}
    \mathcal{B}^{-1/2} & 0 \\
    \mathcal{A}\mathcal{B}^{-1/2} & \mathcal{B}^{1/2}
  \end{bmatrix},
\end{equation*}
one sees that $(Q_{0},P_{0})$ satisfies \eqref{eq:Q_P-Hagedorn} as well as $P_{0}Q_{0}^{-1} = \mathcal{A} + \rmi\mathcal{B}$.
However, setting
\begin{equation*}
  \begin{bmatrix}
    \Re Q & \Im Q \\
    \Re P & \Im P
  \end{bmatrix}
  =
  \begin{bmatrix}
      \mathcal{B}^{-1/2} & 0 \\
      \mathcal{A}\mathcal{B}^{-1/2} & \mathcal{B}^{1/2}
    \end{bmatrix}
    \begin{bmatrix}
      U  & V \\
      -V & U
    \end{bmatrix}
    =
    \begin{bmatrix}
     \mathcal{B}^{-1/2} U & \mathcal{B}^{-1/2} V \\
     \mathcal{A}\mathcal{B}^{-1/2} U - \mathcal{B}^{1/2} V & \mathcal{A}\mathcal{B}^{-1/2} V + \mathcal{B}^{1/2} U
    \end{bmatrix}
\end{equation*}
for any $U + {\rm i}V \in \U(d)$ (the unitary group of degree $d$) would do as well:
$(Q,P)$ again satisfies \eqref{eq:Q_P-Hagedorn} as well as $P Q^{-1} = \mathcal{A} + \rmi\mathcal{B}$.
Therefore one has
\begin{equation}
  \label{eq:Q_P-A_B}
  Q = \mathcal{B}^{-1/2} \mathcal{U},
  \qquad
  P = (\mathcal{A} + \rmi\mathcal{B}) \mathcal{B}^{-1/2} \mathcal{U},
\end{equation}
where $\mathcal{U} \defeq U + {\rm i}V \in \U(d)$.
This is because $\Sigma_{d}$ is actually the homogeneous space $\Sp(2d,\R)/\U(d)$; see, e.g., \citet{Si1943}, \citet[Section~4.5]{Fo1989}, \citet[Exercise~2.28 on p.~48]{McSa1999}, and \citet{Oh2015c} for details.

\subsection{The Hagedorn Wave Packets}
\label{ssec:Hagedorn}
Upon normalizing and getting rid of the phase factor in \eqref{eq:chi_0-Hagedorn}, we have the ground state of the Hagedorn wave packets:
\begin{equation*}
  \varphi_{0}(q,p,Q,P; x) = \frac{(\det Q)^{-1/2}}{(\pi\hbar)^{d/4}} \exp\braces{ \frac{{\rm i}}{\hbar}\brackets{ \frac{1}{2}(x - q)^{T}P Q^{-1}(x - q) + p \cdot (x - q) } }.
\end{equation*}
\citet{Ha1998}\footnote{\citet{Ha1998} uses parameters $A, B \in \Mat_{d}(\C)$, which are related to $Q$ and $P$ as $A = Q$ and $B = -\rmi P$.} came up with the ladder operators
\begin{align*}
  \mathscr{L}(q,p,Q,P) &= -\frac{\rmi}{\sqrt{2\hbar}} \brackets{ P^{T}(\hat{x} - q) - Q^{T}(\hat{p} - p) }, \\
  \mathscr{L}^{*}(q,p,Q,P) &= \frac{\rmi}{\sqrt{2\hbar}} \brackets{ P^{*}(\hat{x} - q) - Q^{*}(\hat{p} - p) },
\end{align*}
that satisfy the same relationships that are satisfied by the ladder operators of the Hermite functions, i.e., 
\begin{equation*}
  \begin{array}{c}
    [\mathscr{L}_{j}(q,p,Q,P), \mathscr{L}_{k}(q,p,Q,P)] = 0,
    \medskip\\
    {[\mathscr{L}_{j}^{*}(q,p,Q,P), \mathscr{L}^{*}_{k}(q,p,Q,P)] = 0},
    \qquad
    [\mathscr{L}_{j}(q,p,Q,P), \mathscr{L}^{*}_{k}(q,p,Q,P)] = \delta_{jk}.
  \end{array}
\end{equation*}
Then these operators are used to define an orthonormal basis $\{ \varphi_{n}(q,p,Q,P;\,\cdot\,) \}_{n \in \N_{0}^{d}}$ for $L^{2}(\R^{d})$ recursively by applying the raising operator $\mathscr{L}^{*}$ repeatedly, i.e., for any multi-index $n = (n_{1}, \dots, n_{d}) \in \N_{0}^{d}$ and $j \in \{1, \dots, d\}$,
\begin{equation*}
  \varphi_{n + e_{j}}(q,p,Q,P;\,\cdot\,) \defeq \frac{1}{\sqrt{n_{j} + 1}}\,\mathscr{L}^{*}_{j}(q,p,Q,P) \varphi_{n}(q,p,Q,P;\,\cdot\,),
\end{equation*}
where $e_{j} \in \R^{d}$ is the unit vector whose $j$-th entry is 1.
One can also show that the lowering operator $\mathscr{L}$ satisfies
\begin{equation*}
  \varphi_{n - e_{j}}(q,p,Q,P;\,\cdot\,) = \frac{1}{\sqrt{n_{j}}}\,\mathscr{L}_{j}(q,p,Q,P) \varphi_{n}(q,p,Q,P;\,\cdot\,).
\end{equation*}

\subsection{Semiclassical Wave Packets}
We would like to use the parametrization $(\mathcal{A},\mathcal{B})$ instead of $(Q,P)$ here.
So we may first rewrite the above ladder operators in terms of $(\mathcal{A},\mathcal{B},\mathcal{U})$ instead of $(Q,P)$ using \eqref{eq:Q_P-A_B}.
But then the resulting operators define ladder operators for an arbitrary $\mathcal{U} \in \U(d)$; hence we set $\mathcal{U} = I_{d}$ to have---with an abuse of notation---the ladder operators
\begin{subequations}
  \label{eq:ladder_operators}
  \begin{align}
    \mathscr{L}(q,p,\mathcal{A},\mathcal{B})
    &\defeq -\frac{\rmi}{\sqrt{2\hbar}}\, \mathcal{B}^{-1/2} \brackets{ (\mathcal{A} + \rmi\mathcal{B}) (\hat{x} - q) - (\hat{p} - p) },\\
    \mathscr{L}^{*}(q,p,\mathcal{A},\mathcal{B})
    &\defeq \frac{\rmi}{\sqrt{2\hbar}}\, \mathcal{B}^{-1/2} \brackets{ (\mathcal{A} - \rmi\mathcal{B})(\hat{x} - q) - (\hat{p} - p) },
    \label{eq:raising_operator}
  \end{align}
\end{subequations}
and generate an orthogonal basis $\{ \chi_{n}(q,p,\mathcal{A},\mathcal{B},\phi,\delta;\,\cdot\,) \}_{n \in \N_{0}^{d}}$ by setting
\begin{equation}
  \label{eq:chi_n-raised}
  \chi_{n + e_{j}}(q,p,\mathcal{A},\mathcal{B},\phi,\delta;\,\cdot\,)
  \defeq \frac{1}{\sqrt{n_{j} + 1}}\,\mathscr{L}^{*}_{j}(q,p,\mathcal{A},\mathcal{B}) \chi_{n}(q,p,\mathcal{A},\mathcal{B},\phi,\delta;\,\cdot\,)
\end{equation}
starting with the ground state~\eqref{eq:chi_0} (without normalization; hence only orthogonal), whereas the lowering operator works as follows:
\begin{equation}
  \label{eq:chi_n-lowered}
  \chi_{n - e_{j}}(q,p,\mathcal{A},\mathcal{B},\phi,\delta;\,\cdot\,)
  = \frac{1}{\sqrt{n_{j}}}\,\mathscr{L}_{j}(q,p,\mathcal{A},\mathcal{B}) \chi_{n}(q,p,\mathcal{A},\mathcal{B},\phi,\delta;\,\cdot\,).
\end{equation}

\begin{remark}
Setting $\mathcal{U} = I_{d}$ has the advantage of parametrizing the wave packets by $(\mathcal{A},\mathcal{B})$ as is done for the Gaussian, but results in a slightly less general form of wave packets than Hagedorn's.
\end{remark}

Note that the norms of these wave packets are all the same because, writing $\chi_{n} = \chi_{n}(y;\,\cdot\,)$ and $\mathscr{L}^{*} = \mathscr{L}^{*}(q,p,\mathcal{A},\mathcal{B})$ for brevity,
\begin{align*}
  \norm{\chi_{n+e_{j}}}^{2}
  &= \frac{1}{n_{j}+1} \ip{ \mathscr{L}^{*}_{j} \chi_{n} }{ \mathscr{L}^{*}_{j} \chi_{n} } \\
  &= \frac{1}{n_{j}+1} \ip{ \mathscr{L}_{j} \mathscr{L}^{*}_{j} \chi_{n} }{ \chi_{n} } \\
  &= \ip{ \chi_{n} }{ \chi_{n} } \\
  &= \norm{ \chi_{n} }^{2},
\end{align*}
and hence we have $\norm{ \chi_{n} }^{2} = \norm{ \chi_{0} }^{2} = \mathcal{N}_{\hbar}(\mathcal{B},\delta)$ for any $n\in\N_{0}^{d}$ by induction, where $\mathcal{N}_{\hbar}$ was defined in \eqref{eq:N}.
As we shall see later in Proposition~\ref{prop:Omega_M-reduced}, $\mathcal{N}_{\hbar}$ is the Noether conserved quantity corresponding to the inherent phase symmetry of our Hamiltonian dynamics; see also Remark~\ref{rem:N} below.
Therefore, if necessary, one may normalize the orthogonal basis $\{ \chi_{n}(q,p,\mathcal{A},\mathcal{B},\phi,\delta;\,\cdot\,) \}_{n \in \N_{0}^{d}}$ to obtain an orthonormal basis just like the Hagedorn wave packets $\{ \varphi_{n}(q,p,Q,P;\,\cdot\,) \}_{n \in \N_{0}^{d}}$.
We will show later in Section~\ref{sec:Symplectic_Reduction} that the normalization corresponds to the symplectic reduction by the phase symmetry with respect to the variable $\phi$.

\section{Embeddings of Semiclassical Wave Packets and Symplectic Structures}
\label{sec:Symplectic_Structures}
Let us write $y \defeq (q, p, \mathcal{A}, \mathcal{B}, \phi, \delta) \in \mathcal{M}$ for short.
Our goal is to come up with the dynamics $y(t)$ of the parameters so that each wave packet $\chi_{n}(y(t);x)$ approximates the solution of the initial value problem of the Schr\"odinger equation~\eqref{eq:SchroedingerEq-coordinates} with the initial condition $\psi(0,x) = \chi_{n}(y(0);x)$.
Particularly we would like to obtain a Hamiltonian dynamics of the parameters $y$ that is naturally related to the Hamiltonian/symplectic structure associated with the Schr\"odinger equation.
This amounts to finding the symplectic structure $\Omega_{\mathcal{M}}^{(n)}$ and Hamiltonian $H^{(n)}$ on $\mathcal{M}$ naturally associated with $\chi_{n}$, and results in the Hamiltonian system defined in terms of $\Omega_{\mathcal{M}}^{(n)}$ and $H^{(n)}$.
Indeed, one can show that this gives the best approximation in some appropriate sense as we shall see below in Section~\ref{ssec:BestApproximation}.

In our previous work~\cite{OhLe2013} on the dynamics of the Gaussian $\chi_{0}$, we followed the approach by \citet{FaLu2006} and \citet[Section~II.1]{Lu2008} and obtained the symplectic structure $\Omega_{\mathcal{M}}^{(0)}$ by regarding the Gaussian wave packet $\chi_{0}(y,\,\cdot\,)$ as the embedding $\iota_{0}\colon\mathcal{M} \hookrightarrow L^{2}(\R^{d})$ defined by $y \mapsto \chi_{0}(y,\,\cdot\,)$; see Section~\ref{ssec:embedding} below for more details.
We would like to generalize the approach to $\chi_{n}$ for an arbitrary $n \in \N_{0}^{d}$.

\subsection{Hamiltonian Formulation of the Schr\"odinger Equation}
\label{ssec:Hamiltonian_Schroedinger}
Let us first briefly review the Hamiltonian formulation of the Schr\"odinger equation following \citet[Section~2.2]{MaRa1999}.
Let $\mathcal{H}$ be a complex Hilbert space---$\mathcal{H} = L^{2}(\R^{d})$ throughout the paper---equipped with a (right-linear) inner product $\ip{\cdot}{\cdot}$.
Then $\mathcal{H}$ is a symplectic vector space with the symplectic structure $\Omega$ defined by
\begin{equation*}
  \Omega(\psi_{1}, \psi_{2}) \defeq 2\hbar \Im\ip{\psi_{1}}{\psi_{2}}.
\end{equation*}
In fact, defining a one-form $\Theta$ on $\mathcal{H}$ by
\begin{equation}
  \label{eq:Theta}
  \Theta(\psi) = -\hbar \Im\ip{\psi}{\mathbf{d}\psi},
\end{equation}
one obtains $\Omega = -\mathbf{d}\Theta$.
Given a Hamiltonian operator $\hat{H}$ on $\mathcal{H}$ (we proceed formally here without specifying the domain of definition of $\hat{H}$), we may write the expectation value of the Hamiltonian $\texval{\hat{H}}\colon \mathcal{H} \to \R$ as
\begin{equation*}
 \texval{\hat{H}}(\psi) \defeq \texval{\psi, \hat{H}\psi}.
\end{equation*}
Now we think of $\texval{\hat{H}}$ as a Hamiltonian {\em function} on the symplectic vector space $\mathcal{H}$, and define the corresponding Hamiltonian vector field $X_{\exval{\hat{H}}}$ on $\mathcal{H}$ by the Hamiltonian system
\begin{equation}
  \label{eq:Schroedinger-HamiltonianSystem}
  {\bf i}_{X_{\texval{\hat{H}}}} \Omega = \mathbf{d}\texval{\hat{H}}.
\end{equation}
Writing the vector field $X_{\exval{\hat{H}}}$ as $X_{\exval{\hat{H}}}(\psi) = (\psi,\dot{\psi}) \in T\mathcal{H} \cong \mathcal{H} \times \mathcal{H}$, one obtains the Schr\"odinger equation
\begin{equation*}
  \dot{\psi} = -\frac{\rmi}{\hbar}\hat{H} \psi.
\end{equation*}
For $\mathcal{H} = L^{2}(\R^{d})$ with the Schr\"odinger operator~\eqref{eq:SchroedingerOp}, the above equation gives \eqref{eq:SchroedingerEq-coordinates}.

\subsection{Embeddings defined by Semiclassical Wave Packets}
\label{ssec:embedding}
We would like to exploit the above Hamiltonian approach to the Schr\"odinger equation in order to formulate Hamiltonian dynamics of the parameters $(q, p, \mathcal{A}, \mathcal{B}, \phi, \delta)$.
First note that the parameters $y = (q, p, \mathcal{A}, \mathcal{B}, \phi, \delta)$ live in the space
\begin{equation*}
 \mathcal{M} \defeq T^{*}\R^{d} \times \Sigma_{d} \times \mathbb{S}^{1} \times \R,
\end{equation*}
which is an even-dimensional manifold for any $d \in \N$ because the (real) dimension of $\Sigma_{d}$ is $d(d+1)$ and hence the dimension of $\mathcal{M}$ is $(d+1)(d+2)$.
Then we may define a family of embeddings of $\mathcal{M}$ to $\mathcal{H} \defeq L^{2}(\R^{d})$ by
\begin{equation}
  \label{eq:iota_n}
  \iota_{n}\colon \mathcal{M} \hookrightarrow \mathcal{H};
  \quad
  \iota_{n}(y) = \chi_{n}(y;\,\cdot\,)
\end{equation}
for any $n \in \N_{0}^{d}$.

Can we naturally induce a symplectic structure on $\mathcal{M}$ from the symplectic structure $\Omega$ on $\mathcal{H}$?
In our previous work~\cite[Proposition~2.1]{OhLe2013}, we reformulated the work of \citet[Section~II.1]{Lu2008} and showed the following:
Let $\iota\colon \mathcal{M} \hookrightarrow \mathcal{H}$ be an embedding of a manifold $\mathcal{M}$ in a complex Hilbert space $\mathcal{H}$ and suppose that $\mathcal{M}$ is equipped with an almost complex structure $J_{y}: T_{y}\mathcal{M} \to T_{y}\mathcal{M}$ that is compatible with the multiplication by the imaginary unit $\rmi$ in $\mathcal{H}$, i.e.,
\begin{equation}
  \label{eq:iota-compatibility}
  T_{y}\iota \circ J_{y} = \rmi\cdot T_{y}\iota
\end{equation}
for any $y \in \mathcal{M}$; then $\mathcal{M}$ is a symplectic manifold with symplectic form defined by the pull-back $\Omega_{\mathcal{M}} \defeq \iota^{*}\Omega$.
In \cite{OhLe2013}, we worked out the Gaussian case, i.e., $\iota = \iota_{0}$ explicitly:
We found that
\begin{multline}
  \label{eq:J}
  J_{y}\parentheses{ \dot{q}, \dot{p}, \dot{\mathcal{A}}, \dot{\mathcal{B}}, \dot{\phi}, \dot{\delta} }
  \\
  = \parentheses{
    \mathcal{B}^{-1}(\mathcal{A}\dot{q} - \dot{p}),\,
    (\mathcal{A}\mathcal{B}^{-1}\mathcal{A} + \mathcal{B})\dot{q} - \mathcal{A}\mathcal{B}^{-1}\dot{p},\,
    -\dot{\mathcal{B}},\,
    \dot{\mathcal{A}},\,
    p^{T}\mathcal{B}^{-1}(\mathcal{A}\dot{q} - \dot{p}) - \dot{\delta},\,
    -p \cdot \dot{q} + \dot{\phi}
  },
\end{multline}
is an almost complex structure that satisfies $T_{y}\iota_{0} \circ J_{y} = \rmi\cdot T_{y}\iota_{0}$, and found the pull-back $\Theta_{\mathcal{M}}^{(0)} \defeq \iota_{0}^{*} \Theta$ of the canonical one-form $\Theta$ in \eqref{eq:Theta}.
Setting $\Omega_{\mathcal{M}}^{(0)} \defeq -\d\Theta_{\mathcal{M}}^{(0)}$ gives a symplectic form on $\mathcal{M}$.

\subsection{Hamiltonian Dynamics as the Best Approximation}
\label{ssec:BestApproximation}
Given an embedding $\iota\colon \mathcal{M} \hookrightarrow \mathcal{H}$ with $\iota(y) = \chi(y;\,\cdot\,)$ satisfying \eqref{eq:iota-compatibility}, one can also define a Hamiltonian function $H$ on $\mathcal{M}$ via the above embedding as $H \defeq \iota^{*}\texval{\hat{H}} = \bigl\langle \chi, \hat{H} \chi \bigr\rangle$.
So one can formulate a Hamiltonian system on $\mathcal{M}$ as $\ins{X_{H}}{\Omega_{\mathcal{M}}} = \d{H}$.
As shown in \citet[Section~II.1.2]{Lu2008} (see also \cite[Proposition~2.4]{OhLe2013}), the Hamiltonian vector field $X_{H}$ gives the best approximation to the vector field $X_{\texval{\hat{H}}}$ of the Schr\"odinger dynamics in the following sense: $X_{H}$ is the least squares approximation---in terms of the norm in $L^{2}(\R^{d})$---among the vector fields on $\mathcal{M}$ to the vector field $X_{\exval{\hat{H}}}$ defined by the Schr\"odinger equation~\eqref{eq:Schroedinger-HamiltonianSystem}.
More specifically, we have, for any $y \in \mathcal{M}$, 
\begin{equation*}
  \| X_{\texval{\hat{H}}}(\iota(y)) - T_{y}\iota(X_H(y)) \|
  \le
  \| X_{\texval{\hat{H}}}(\iota(y)) - T_{y}\iota(V_{y}) \|
\end{equation*}
for any $V_{y} \in T_{y}\mathcal{M}$, where the equality holds if and only if $V_{y} = X_{H}(y)$; see Fig.~\ref{fig:BestApproximation}.

\begin{figure}
  \centering
  \includegraphics[width=.6\linewidth]{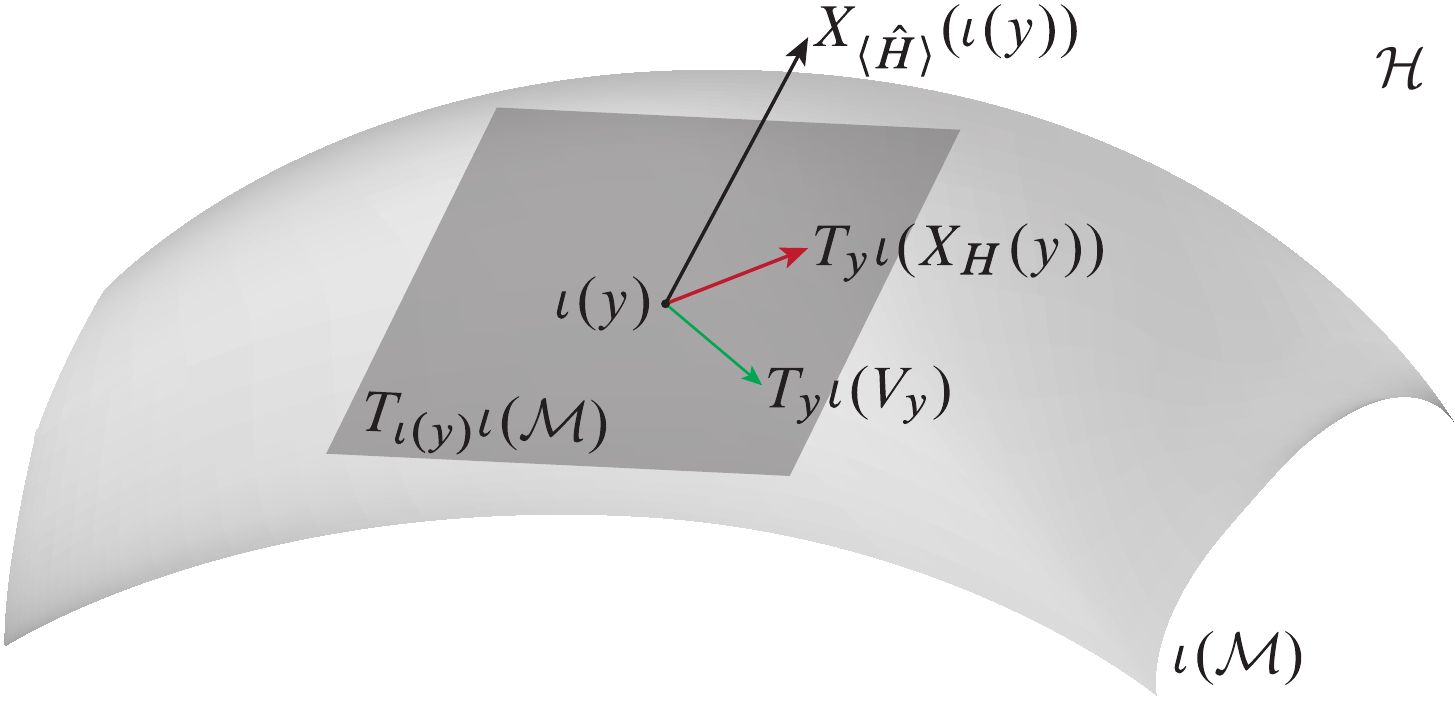}
  \caption{The Hamiltonian vector field $X_{H}$ gives the best approximation among the vector fields on $\mathcal{M}$ to the Schr\"odinger dynamics $X_{\texval{\hat{H}}}$.}
  \label{fig:BestApproximation}
\end{figure}

\subsection{Semiclassical Symplectic Structures}
We would like to apply the above approach to the embedding $\iota_{n}\colon \mathcal{M} \hookrightarrow \mathcal{H}$ of the semiclassical wave packets $\chi_{n}(y;\,\cdot\,)$ with an arbitrary $n \in \N_{0}^{d}$.
As the first step, let us find the symplectic form $\Omega^{(n)}_{\mathcal{M}} \defeq \iota_{n}^{*}\Omega$.
As mentioned earlier, a concrete expression for $\chi_{n}$ with general $n \in \N_{0}^{d}$ is essentially out of reach, and so it is not feasible to calculate $\Theta_{\mathcal{M}}^{(n)} \defeq \iota_{n}^{*} \Theta$ directly.
A way around it is to find a recurrence relation between the canonical one-forms $\{ \Theta_{\mathcal{M}}^{(n)} \}_{n\in\N_{0}^{d}}$.
\begin{proposition}
  \label{prop:Omega_M}
  Let $\iota_{n}\colon \mathcal{M} \hookrightarrow \mathcal{H} = L^{2}(\R^{d})$ be the embedding~\eqref{eq:iota_n} defined by the semiclassical wave packet $\chi_{n}(y;\,\cdot\,)$ for an arbitrary multi-index $n = (n_{1}, \dots, n_{d}) \in \N_{0}^{d}$.
  Then the almost complex structure $J$ in \eqref{eq:J} satisfies $T_{y}\iota_{n} \circ J_{y} = \rmi\cdot T_{y}\iota_{n}$ for any $y \in \mathcal{M}$ and hence the pull-back $\Omega_{\mathcal{M}}^{(n)} \defeq \iota_{n}^{*}\Omega$ defines a symplectic structure on $\mathcal{M}$, and is given by $\Omega_{\mathcal{M}}^{(n)} = -\mathbf{d}\Theta_{\mathcal{M}}^{(n)}$ with
  \begin{equation}
    \label{eq:Theta_M}
    \Theta_{\mathcal{M}}^{(n)} \defeq \iota_{n}^{*}\Theta
    = \mathcal{N}_{\hbar}(\mathcal{B},\delta) \parentheses{
      p_{i}\,\mathbf{d}q_{i} - \frac{\hbar}{4}\tr\parentheses{ (\mathcal{B}^{(n)})^{-1} \mathbf{d}\mathcal{A} } - \mathbf{d}\phi
    },
  \end{equation}
  where $\mathcal{B}^{(n)}$ is the $d \times d$ positive-definite matrix defined as
  \begin{equation}
    \label{eq:mathcalB_n}
    \mathcal{B}^{(n)} \defeq \mathcal{B}^{1/2} (\Lambda^{(n)})^{-1} \mathcal{B}^{1/2},
  \end{equation}
  and $\Lambda^{(n)}$ is the $d \times d$ diagonal matrix defined as
  \begin{equation}
    \label{eq:Lambda_n}
    \Lambda^{(n)} \defeq \diag(2n_{1}+1, \dots, 2n_{d}+1).
  \end{equation}
  More explicitly,
  \begin{equation}
    \label{eq:Omega_M}
    \begin{split}
      \Omega_{\mathcal{M}}^{(n)}
      =\mathcal{N}_{\hbar}(\mathcal{B},\delta) \biggl\{&
      \mathbf{d}q_{i} \wedge \mathbf{d}p_{i}
      - \frac{p_{i}}{2} \mathbf{d}q_{i} \wedge \tr\parentheses{ (\mathcal{B}^{(n)})^{-1}\mathbf{d}\mathcal{B}^{(n)} }
      - \frac{2p_{i}}{\hbar}\mathbf{d}q_{i} \wedge \mathbf{d}\delta \\
      &- \frac{\hbar}{4}\d\mathcal{A}_{ij} \wedge \d(\mathcal{B}^{(n)})^{-1}_{ij} + \frac{\hbar}{8}\tr\parentheses{ (\mathcal{B}^{(n)})^{-1}\mathbf{d}\mathcal{A} } \wedge \tr\parentheses{ (\mathcal{B}^{(n)})^{-1}\d\mathcal{B}^{(n)} } \\
      &+ \frac{1}{2}\tr\parentheses{ (\mathcal{B}^{(n)})^{-1}\d\mathcal{A} } \wedge \d\delta
      - \frac{1}{2}\tr\parentheses{ (\mathcal{B}^{(n)})^{-1}\d\mathcal{B}^{(n)} } \wedge \d\phi + \frac{2}{\hbar} \mathbf{d}\phi \wedge \mathbf{d}\delta
      \biggr\}.
    \end{split}
  \end{equation}
\end{proposition}

\begin{proof}
  Let us first show that $T_{y}\iota_{n} \circ J_{y} = \rmi\cdot T_{y}\iota_{n}$ for any $n \in \N_{0}^{d}$ by induction.
  One can check that it holds for $n = 0$ by direct calculations.
  Now let $n \in \N_{0}^{d}$ and suppose that $\iota_{n}$ satisfies $T_{y}\iota_{n} \circ J_{y} = \rmi\cdot T_{y}\iota_{n}$, and let $e_{j} \in \R^{d}$ be the unit vector with 1 in the $j$-th component with $j \in \{1, \dots, d\}$.
  Then \eqref{eq:chi_n-raised} implies that the embeddings $\iota_{n}$ and $\iota_{n+e_{j}}$ are related as
  \begin{equation*}
    \iota_{n+e_{j}} = \frac{1}{\sqrt{n_{j} + 1}}\mathscr{L}_{j}^{*} \circ \iota_{n}
    \quad
    \text{or}
    \quad
    \begin{tikzcd}[row sep=7ex, column sep=7ex]
      & L^{2}(\R^{d}) \\
      \mathcal{M} \arrow[r,hook,swap,"\iota_{n}"] \arrow[ru,hook,"\iota_{n+e_{j}}"] & L^{2}(\R^{d}) \arrow[u,swap,"\frac{1}{\sqrt{n_{j} + 1}}\mathscr{L}_{j}^{*}"]
    \end{tikzcd}.
  \end{equation*}
  However, since $\mathscr{L}_{j}^{*}$ is a linear operator, we have, for any $y \in \mathcal{M}$,
  \begin{equation*}
    T_{y}\iota_{n+e_{j}} = \frac{1}{\sqrt{n_{j} + 1}}\mathscr{L}_{j}^{*} \circ T_{y}\iota_{n}.
  \end{equation*}
  But then this implies that
  \begin{align*}
    T_{y}\iota_{n+e_{j}} \circ J_{y}
    &= \frac{1}{\sqrt{n_{j} + 1}}\mathscr{L}_{j}^{*} \circ T_{y}\iota_{n} \circ J_{y} \\
    &= \frac{\rmi}{\sqrt{n_{j} + 1}}\mathscr{L}_{j}^{*} \circ T_{y}\iota_{n} \\
    &= \rmi\cdot T_{y}\iota_{n+e_{j}}.
  \end{align*}

  The expression \eqref{eq:Theta_M} follows from the following recurrence relation that holds between the one-forms $\{ \Theta_{\mathcal{M}}^{(n)} \}_{n \in \N_{0}^{d}}$:
  Let $n \in \N_{0}^{d}$ and $e_{j}$ be as above; then, as we shall prove later,
  \begin{equation}
    \label{eq:rec_rel-Theta}
    \Theta_{\mathcal{M}}^{(n+e_{j})} = \Theta_{\mathcal{M}}^{(n)} - \frac{\hbar}{2} \mathcal{N}_{\hbar}(\mathcal{B},\delta) \parentheses{ \mathcal{B}^{-1/2}\,\d\mathcal{A}\,\mathcal{B}^{-1/2} }_{jj},
  \end{equation}
  where summation on the index $j$ is {\em not} assumed on the right-hand side, i.e., it is the $(j,j)$-entry of the matrix $\mathcal{B}^{-1/2}\,\d\mathcal{A}\,\mathcal{B}^{-1/2}$.
  But then direct calculations yield, as is done in \cite{OhLe2013},
  \begin{equation*}
    \Theta_{\mathcal{M}}^{(0)} \defeq \iota_{0}^{*}\Theta = \mathcal{N}_{\hbar}(\mathcal{B},\delta) \parentheses{
      p_{i}\,\mathbf{d}q_{i} - \frac{\hbar}{4}\tr\parentheses{ \mathcal{B}^{-1/2}\,\mathbf{d}\mathcal{A}\,\mathcal{B}^{-1/2} } - \mathbf{d}\phi
    },
  \end{equation*}
  and hence we obtain \eqref{eq:Theta_M} as follows:
  \begin{align*}
    \Theta_{\mathcal{M}}^{(n)}
    &= \Theta_{\mathcal{M}}^{(0)} - \frac{\hbar}{2} \mathcal{N}_{\hbar}(\mathcal{B},\delta) \sum_{j=1}^{d} n_{j} (\mathcal{B}^{-1/2}\,\d\mathcal{A}\,\mathcal{B}^{-1/2})_{jj} \\
    &= \mathcal{N}_{\hbar}(\mathcal{B},\delta) \parentheses{
      p_{i}\,\mathbf{d}q_{i} - \frac{\hbar}{4}\tr\parentheses{ \mathcal{B}^{-1/2}\Lambda^{(n)}\mathcal{B}^{-1/2}\,\mathbf{d}\mathcal{A} } - \mathbf{d}\phi
      } \\
    &= \mathcal{N}_{\hbar}(\mathcal{B},\delta) \parentheses{
      p_{i}\,\mathbf{d}q_{i} - \frac{\hbar}{4}\tr\parentheses{ (\mathcal{B}^{(n)})^{-1} \mathbf{d}\mathcal{A} } - \mathbf{d}\phi
      },
  \end{align*}
  where $\mathcal{B}^{(n)}$ is defined in \eqref{eq:mathcalB_n}.
  Then we have
  \begin{align*}
    \Omega_{\mathcal{M}}^{(n)}
    &= \iota_{n}^{*}\Omega \\
    &= -\iota_{n}^{*}\d\Theta \\
    &= -\d\iota_{n}^{*}\Theta \\
    &= -\d\Theta_{\mathcal{M}}^{(n)},
  \end{align*}
  and the expression~\eqref{eq:Omega_M} for $\Omega_{\mathcal{M}}^{(n)}$ follows from tedious but straightforward calculations; note that 
  \begin{align*}
    \d\mathcal{N}_{\hbar}(\mathcal{B},\delta)
    &= \mathcal{N}_{\hbar}(\mathcal{B},\delta) \parentheses{
      -\frac{1}{2}\tr\parentheses{ \mathcal{B}^{-1}\d\mathcal{B} }
      -\frac{2}{\hbar}\d\delta
    } \\
    &= \mathcal{N}_{\hbar}(\mathcal{B},\delta) \parentheses{
      -\frac{1}{2}\tr\parentheses{ (\mathcal{B}^{(n)})^{-1}\d\mathcal{B}^{(n)} }
      -\frac{2}{\hbar}\d\delta
    }.
  \end{align*}
  So it remains to prove the recurrence relation~\eqref{eq:rec_rel-Theta}.
  Using \eqref{eq:chi_n-raised}, we have
  \begin{align*}
    \Theta_{\mathcal{M}}^{(n+e_{j})}
    &= \iota_{n+e_{j}}^{*}\Theta \\
    &= -\hbar \Im\ip{ \chi_{n+e_{j}} }{ D_{y}\chi_{n+e_{j}} } \cdot \d{y} \\
    &= -\frac{\hbar}{n_{j}+1} \Im\ip{ \mathscr{L}^{*}_{j} \chi_{n} }{ D_{y}\parentheses{ \mathscr{L}^{*}_{j} \chi_{n} } } \cdot \d{y} \\
    &= -\frac{\hbar}{n_{j}+1} \Bigl(
      \Im\ip{ \mathscr{L}^{*}_{j} \chi_{n} }{  D_{y}\mathscr{L}^{*}_{j} \chi_{n} }
       + \Im\ip{ \mathscr{L}^{*}_{j} \chi_{n} }{ \mathscr{L}^{*}_{j} D_{y}\chi_{n} }
      \Bigr) \cdot \d{y},
  \end{align*}
  where again no summation is assumed on $j$.
  Using the properties \eqref{eq:chi_n-raised} and \eqref{eq:chi_n-lowered} of the ladder operators, we have
  \begin{align*}
    \Im\ip{ \mathscr{L}^{*}_{j} \chi_{n} }{ \mathscr{L}^{*}_{j} D_{y}\chi_{n} } \cdot \d{y}
    &= \Im\ip{ \mathscr{L}_{j} \mathscr{L}^{*}_{j} \chi_{n} }{ D_{y}\chi_{n} } \cdot \d{y} \\
    &= (n_{j}+1) \Im\ip{ \chi_{n} }{ D_{y}\chi_{n} } \cdot \d{y} \\
    &= -\frac{n_{j}+1}{\hbar} \Theta_{\mathcal{M}}^{(n)}.
  \end{align*}
  Therefore we obtain the recurrence relation
  \begin{equation}
    \label{eq:rec_rel-Theta-pre}
    \Theta_{\mathcal{M}}^{(n+e_{j})}
    = \Theta_{\mathcal{M}}^{(n)}
      - \frac{\hbar}{n_{j}+1}
      \Im\ip{ \mathscr{L}^{*}_{j} \chi_{n} }{ D_{y}\mathscr{L}^{*}_{j} \chi_{n} } \cdot \d{y}.
  \end{equation}
  Let us evaluate the second term on the right-hand side.
  Taking the derivatives of \eqref{eq:raising_operator} with respect to $(q,p)$, we have
  \begin{equation*}
    D_{q_{l}}\mathscr{L}^{*}_{j} = -\frac{\rmi}{\sqrt{2\hbar}}\,\mathcal{B}^{-1/2}_{jk}(\mathcal{A} - \rmi\mathcal{B})_{kl},
    \qquad
    D_{p_{l}}\mathscr{L}^{*}_{j} = \frac{\rmi}{\sqrt{2\hbar}}\,\mathcal{B}^{-1/2}_{jl},
  \end{equation*}
  and hence
  \begin{align*}
    \ip{ \mathscr{L}^{*}_{j} \chi_{n} }{ D_{q_{l}}\mathscr{L}^{*}_{j} \chi_{n} }
    &= -\frac{\rmi}{\sqrt{2\hbar}}\,\mathcal{B}^{-1/2}_{jk}(\mathcal{A} - \rmi\mathcal{B})_{kl} \ip{ \mathscr{L}^{*}_{j} \chi_{n} }{ \chi_{n} }, \\
    \ip{ \mathscr{L}^{*}_{j} \chi_{n} }{ D_{p_{l}}\mathscr{L}^{*}_{j} \chi_{n} }
    &= \frac{\rmi}{\sqrt{2\hbar}}\,\mathcal{B}^{-1/2}_{jl} \ip{ \mathscr{L}^{*}_{j} \chi_{n} }{ \chi_{n} }.
  \end{align*}
  However, they both vanish due to the orthogonality of the basis $\{ \chi_{n} \}_{n\in\N_{0}^{d}}$:
  \begin{equation*}
    \ip{ \mathscr{L}^{*}_{j} \chi_{n} }{ \chi_{n} }
    = \sqrt{n_{j}+1}\,\tip{ \chi_{n+{e_{j}}} }{ \chi_{n} } = 0.
  \end{equation*}
  On the other hand, taking the derivatives of \eqref{eq:raising_operator} with respect to $\mathcal{A}$ and $\mathcal{B}$, we have
  \begin{align*}
    D_{\mathcal{A}_{lr}}\mathscr{L}^{*}_{j}
    &= \frac{\rmi}{2\sqrt{2\hbar}}\,\parentheses{
      \mathcal{B}^{-1/2}_{jl}(\hat{x} - q)_{r} + \mathcal{B}^{-1/2}_{jr}(\hat{x} - q)_{l}
      } \\
    &= \frac{\rmi}{4}\, \parentheses{ \mathcal{B}^{-1/2}_{jl} \mathcal{B}^{-1/2}_{rk} + \mathcal{B}^{-1/2}_{jr} \mathcal{B}^{-1/2}_{lk} } ( \mathscr{L}_{k} + \mathscr{L}^{*}_{k} )
  \end{align*}
  and
  \begin{align*}
    D_{\mathcal{B}_{lr}}\mathscr{L}^{*}_{j}
    &= \frac{\rmi}{\sqrt{2\hbar}}\, D_{\mathcal{B}_{lr}}\mathcal{B}^{-1/2}_{js} \bigl(
      (\mathcal{A} - \rmi\mathcal{B}) (\hat{x} - q) - (\hat{p} - p)
      \bigr)_{s}
      + \frac{1}{2\sqrt{2\hbar}}\, \parentheses{ \mathcal{B}^{-1/2}_{jl} (\hat{x} - q)_{r} + \mathcal{B}^{-1/2}_{jr} (\hat{x} - q)_{l} }\\
    &= D_{\mathcal{B}_{lr}}\mathcal{B}^{-1/2}_{js} \mathcal{B}^{1/2}_{su} \mathscr{L}^{*}_{u}
      + \frac{1}{4} \parentheses{ \mathcal{B}^{-1/2}_{jl} \mathcal{B}^{-1/2}_{ru} + \mathcal{B}^{-1/2}_{jr} \mathcal{B}^{-1/2}_{lu} } ( \mathscr{L}_{u} + \mathscr{L}^{*}_{u} ),
  \end{align*}
  where we used \eqref{eq:raising_operator} as well as the following identity that follows from \eqref{eq:ladder_operators}:
  \begin{equation}
    \label{eq:x-q_in_As}
    \hat{x} - q = \sqrt{\frac{\hbar}{2}}\,\mathcal{B}^{-1/2} ( \mathscr{L} + \mathscr{L}^{*} ).
  \end{equation}
  So we have 
  \begin{align*}
    \ip{ \mathscr{L}^{*}_{j} \chi_{n} }{ D_{\mathcal{A}_{lr}}\mathscr{L}^{*}_{j} \chi_{n} }
    &= \frac{\rmi}{4}\,\parentheses{ \mathcal{B}^{-1/2}_{jl} \mathcal{B}^{-1/2}_{rk} + \mathcal{B}^{-1/2}_{jr} \mathcal{B}^{-1/2}_{lk} } \parentheses{ 
      \ip{ \mathscr{L}^{*}_{j} \chi_{n} }{ \mathscr{L}_{k} \chi_{n} }
      + \ip{ \mathscr{L}^{*}_{j} \chi_{n} }{ \mathscr{L}^{*}_{k} \chi_{n} }
      } \\
    &= (n_{j}+1)\frac{\rmi}{4}\,\mathcal{N}_{\hbar}(\mathcal{B},\delta)\,\parentheses{ \mathcal{B}^{-1/2}_{jl} \mathcal{B}^{-1/2}_{rj} + \mathcal{B}^{-1/2}_{jr} \mathcal{B}^{-1/2}_{lj} }
  \end{align*}
  with no summation on the index $j$, since 
  \begin{equation}
    \label{eq:inner_products}
    \ip{ \mathscr{L}^{*}_{j} \chi_{n} }{ \mathscr{L}_{k} \chi_{n} } = 0,
    \qquad
    \ip{ \mathscr{L}^{*}_{j} \chi_{n} }{ \mathscr{L}^{*}_{k} \chi_{n} } = \delta_{jk}\, (n_{j}+1) \mathcal{N}_{\hbar}(\mathcal{B},\delta)
  \end{equation}
  due to \eqref{eq:chi_n-raised} and \eqref{eq:chi_n-lowered}.
  On the other hand, we see that the term $\tip{ \mathscr{L}^{*}_{j} \chi_{n} }{ D_{\mathcal{B}_{lr}}\mathscr{L}^{*}_{j} \chi_{n} }$ is real and hence does not contribute to \eqref{eq:rec_rel-Theta-pre}.
  As a result, we obtain
  \begin{align*}
    \Im\ip{ \mathscr{L}^{*}_{j} \chi_{n} }{ D_{y}\mathscr{L}^{*}_{j} \chi_{n} } \cdot \d{y}
    &= \Im\ip{ \mathscr{L}^{*}_{j} \chi_{n} }{ D_{\mathcal{A}_{lr}}\mathscr{L}^{*}_{j} \chi_{n} } \d\mathcal{A}_{lr} \\
    &= \frac{n_{j}+1}{2}\,\mathcal{N}_{\hbar}(\mathcal{B},\delta) (\mathcal{B}^{-1/2}\,\d\mathcal{A}\,\mathcal{B}^{-1/2})_{jj},
  \end{align*}
  and hence substituting this into \eqref{eq:rec_rel-Theta-pre} yields the recurrence relation \eqref{eq:rec_rel-Theta}.
\end{proof}

\section{Hamiltonian Dynamics of Semiclassical Wave Packets}
\label{sec:Hamiltonian_Dynamics_of_SWP}
Now that we have the symplectic forms $\{ \Omega_{\mathcal{M}}^{(n)} \}_{n\in\N_{0}^{d}}$, it remains to find the Hamiltonians $\{ H^{(n)} \}_{n\in\N_{0}^{d}}$ that correspond to the semiclassical wave packets $\{ \chi_{n} \}_{n\in\N_{0}^{d}}$ in order to formulate Hamiltonian dynamics for them.

In our previous work~\cite{OhLe2013}, we found the Hamiltonian $H^{(0)}$ corresponding to the Gaussian $\chi_{0}$ via an asymptotic expansion of the pull-back of the expectation value of the Hamiltonian operator $\hat{H}$, i.e.,
\begin{equation}
  \label{eq:H^0}
  \texval{\hat{H}}^{(0)} \defeq \iota_{0}^{*}\texval{\hat{H}}
  = \bigl\langle \chi_{0}, \hat{H} \chi_{0} \bigr\rangle \\
  = H^{(0)} + O(\hbar^{2}).
\end{equation}
Then the Hamiltonian system $\ins{X_{H^{(0)}}}{\Omega^{(0)}} = \d{H^{(0)}}$ yields \eqref{eq:Heller} with the second equation replaced by \eqref{eq:pdot-chi_0}.
In this section, we would like to generalize this result to $\chi_{n}$ with an arbitrary $n\in\N_{0}^{d}$.

\subsection{Semiclassical Hamiltonians}
Let us find an asymptotic expansion for the expectation value
\begin{equation*}
  \texval{\hat{H}}^{(n)} \defeq \iota_{n}^{*}\texval{\hat{H}}
  = \bigl\langle \chi_{n}, \hat{H} \chi_{n} \bigr\rangle
\end{equation*}
of the Schr\"odinger operator $\hat{H}$ in \eqref{eq:SchroedingerOp} with respect to the semiclassical wave packet $\chi_{n}$.
We will evaluate the kinetic and potential parts of $\{ H^{(n)} \}_{n\in\N_{0}^{d}}$ separately:
It turns out that the kinetic part can be found again via a recurrence relation by induction on $n$, whereas the potential part can be evaluated directly as an asymptotic expansion in $\hbar$ for any $n \in \N_{0}^{d}$ under a reasonable technical assumption on the potential $V$.

\begin{proposition}
  \label{prop:H_M}
  Suppose that the potential $V$ is in $C^{3}(\R^{d})$ and that there exist $C_{1}, C_{2}, M_{1} \in \R$ such that $C_{1} \le V(x)$ and for any $\alpha \in \N_{0}^{d}$ with $|\alpha| = 3$,
  \begin{equation}
    \label{eq:assumption_on_D3V}
    |D^{\alpha}V(x)| \le C_{2} \exp(M_{1}|x|^{2}).
  \end{equation}
  Then the expectation value $\texval{\hat{H}}^{(n)}$ for each $n\in\N_{0}^{d}$ has the asymptotic expansion
  \begin{equation}
    \label{eq:exvalH-asymptotic}
    \texval{\hat{H}}^{(n)} = H^{(n)} + \mathcal{N}_{\hbar}(\mathcal{B},\delta)\,O(\hbar^{3/2}),
  \end{equation}
  where $H^{(n)}\colon \mathcal{M} \to \R$ is defined as
  \begin{equation}
    \label{eq:H_M}
    H^{(n)} \defeq \mathcal{N}_{\hbar}(\mathcal{B},\delta) \Biggl\{
    \frac{p^{2}}{2m} + \frac{\hbar}{4m} \tr\!\parentheses{ (\mathcal{B}^{(n)})^{-1}(\mathcal{A}^{2} + \mathcal{B}^{2}) } \\
    + V(q) + \frac{\hbar}{4} \tr\!\parentheses{  (\mathcal{B}^{(n)})^{-1} D^{2}V(q) }
    \Biggr\},
  \end{equation}
  with $\mathcal{B}^{(n)}$ defined in \eqref{eq:mathcalB_n}.
\end{proposition}

\begin{remark}
  \label{rem:N}
  The quantity $\mathcal{N}_{\hbar}(\mathcal{B},\delta) = \norm{ \chi_{n}(y;\,\cdot\,) }^{2}$ depends on $\hbar$ as shown in \eqref{eq:N}.
  However, as we shall see later in Proposition~\ref{prop:Omega_M-reduced}, $\mathcal{N}_{\hbar}(\mathcal{B},\delta)$ is conserved along the Hamiltonian dynamics of the parameters $y = (q,p,\mathcal{A},\mathcal{B},\phi,\delta)$ that we derive later.
  Therefore, upon normalizing the wave packet $\chi_{n}(y;\,\cdot\,)$ in the initial condition by setting $\mathcal{N}_{\hbar}(\mathcal{B},\delta) = 1$, it stays so all time; hence we may assume $\mathcal{N}_{\hbar}(\mathcal{B},\delta) = O(1)$.
\end{remark}

\begin{remark}
  The error term becomes $O(\hbar^{2})$ if $V$ is $C^{4}(\R^{d})$ and assuming \eqref{eq:assumption_on_D3V} for $|\alpha| = 4$.
  In fact, the asymptotic expansion in \eqref{eq:H^0} assumes that $V$ is smooth and satisfies a condition similar to \eqref{eq:assumption_on_D3V}; see \cite[Proposition~7.1]{OhLe2013}.
\end{remark}

\begin{proof}
  Note first that the assumption that $V$ is bounded from below guarantees that the Schr\"odinger operator~\eqref{eq:SchroedingerOp} is essentially self-adjoint.
  Let us split the expectation value of the Hamiltonian into the kinetic and potential parts, i.e.,
  \begin{equation*}
    \texval{\hat{H}}^{(n)} = \texval{\hat{T}}^{(n)} + \exval{V}^{(n)}
  \end{equation*}
  with $\hat{T} \defeq \hat{p}^{2}/(2m)$, and first evaluate the kinetic part.
  We see that
  \begin{align*}
    \texval{\hat{T}}^{(n+e_{j})}
    &= \bigl\langle \chi_{n+e_{j}}, \hat{T} \chi_{n+e_{j}} \bigr\rangle \\
    &= \frac{1}{n_{j}+1} \bigl\langle \mathscr{L}^{*}_{j} \chi_{n}, \hat{T} \mathscr{L}^{*}_{j} \chi_{n} \bigr\rangle \\
    &= \frac{1}{n_{j}+1} \parentheses{
      \bigl\langle \mathscr{L}^{*}_{j} \chi_{n}, \bigl[\hat{T}, \mathscr{L}^{*}_{j}\bigr] \chi_{n} \bigr\rangle
      + \bigl\langle \mathscr{L}^{*}_{j} \chi_{n}, \mathscr{L}^{*}_{j} \hat{T} \chi_{n} \bigr\rangle
      },
  \end{align*}
  but then
  \begin{align*}
    \bigl\langle \mathscr{L}^{*}_{j} \chi_{n}, \mathscr{L}^{*}_{j} \hat{T} \chi_{n} \bigr\rangle
    &= \bigl\langle \mathscr{L}_{j} \mathscr{L}^{*}_{j} \chi_{n}, \hat{T} \chi_{n} \bigr\rangle \\
    &= (n_{j}+1) \bigl\langle \chi_{n}, \hat{T} \chi_{n} \bigr\rangle \\
    &= (n_{j}+1) \texval{\hat{T}}^{(n)},
  \end{align*}
  and hence we have the recurrence relation
  \begin{equation*}
    \texval{\hat{T}}^{(n+e_{j})} = \texval{\hat{T}}^{(n)} + \frac{1}{n_{j}+1} \bigl\langle \mathscr{L}^{*}_{j} \chi_{n}, \bigl[\hat{T}, \mathscr{L}^{*}_{j}\bigr] \chi_{n} \bigr\rangle.
  \end{equation*}
  
  Let us evaluate the second term on the right-hand side.
  It is straightforward to see that, using \eqref{eq:raising_operator},
  \begin{equation*}
    \bigl[ \hat{T}, \mathscr{L}^{*} \bigr]
    = \frac{1}{m}\sqrt{\frac{\hbar}{2}}\, \mathcal{B}^{-1/2} (\mathcal{A} - \rmi\mathcal{B}) \hat{p},
  \end{equation*}
  and hence we have
  \begin{align*}
    \bigl\langle \mathscr{L}^{*}_{j} \chi_{n}, \bigl[\hat{T}, \mathscr{L}^{*}_{j}\bigr] \chi_{n} \bigr\rangle
    &= \frac{1}{m}\sqrt{\frac{\hbar}{2}}\, \mathcal{B}^{-1/2}_{jk}
      (\mathcal{A} - \rmi\mathcal{B})_{kl} \bigip{ \mathscr{L}^{*}_{j} \chi_{n} }{ \hat{p}_{l} \chi_{n} }.
  \end{align*}
  It is easy to see from the definition \eqref{eq:ladder_operators} of the ladder operators (see also \eqref{eq:x-q_in_As}) that
  \begin{equation*}
    \hat{p} - p = \sqrt{\frac{\hbar}{2}}\,\bigl(
    (\mathcal{A} - \rmi\mathcal{B}) \mathcal{B}^{-1/2} \mathscr{L}
    + (\mathcal{A} + \rmi\mathcal{B}) \mathcal{B}^{-1/2} \mathscr{L}^{*}
    \bigr),
  \end{equation*}
  and so
  \begin{align*}
    \bigip{ \mathscr{L}^{*}_{j} \chi_{n} }{ \hat{p}_{l} \chi_{n} }
    &= \bigip{ \mathscr{L}^{*}_{j} \chi_{n} }{ \chi_{n} } p_{l} \\
    &\quad + \sqrt{\frac{\hbar}{2}}\,\bigl(
      (\mathcal{A} - \rmi\mathcal{B})_{lr} \mathcal{B}^{-1/2}_{rs} \bigip{ \mathscr{L}^{*}_{j} \chi_{n} }{ \mathscr{L}_{s} \chi_{n} }
      + (\mathcal{A} + \rmi\mathcal{B})_{lr} \mathcal{B}^{-1/2}_{rs} \bigip{ \mathscr{L}^{*}_{j} \chi_{n} }{ \mathscr{L}^{*}_{s} \chi_{n} }
      \bigr) \\
    &= (n_{j}+1)\sqrt{\frac{\hbar}{2}}\,\mathcal{N}_{\hbar}(\mathcal{B},\delta)\,(\mathcal{A} + \rmi\mathcal{B})_{lr} \mathcal{B}^{-1/2}_{rj}
  \end{align*}
  because, due to the properties~\eqref{eq:chi_n-raised} and \eqref{eq:chi_n-lowered} of the ladder operators and the orthogonality of $\{ \chi_{n} \}_{n\in\N_{0}^{d}}$,
  \begin{equation*}
    \bigip{ \mathscr{L}^{*}_{j} \chi_{n} }{ \chi_{n} } = 0,
    \qquad
    \bigip{ \mathscr{L}^{*}_{j} \chi_{n} }{ \mathscr{L}_{s} \chi_{n} } = 0,
    \qquad
    \bigip{ \mathscr{L}^{*}_{j} \chi_{n} }{ \mathscr{L}^{*}_{s} \chi_{n} }
    = \delta_{js}\,(n_{j}+1)\,\mathcal{N}_{\hbar}(\mathcal{B},\delta).
  \end{equation*}
  Therefore,
  \begin{align*}
    \bigl\langle \mathscr{L}^{*}_{j} \chi_{n}, \bigl[\hat{T}, \mathscr{L}^{*}_{j}\bigr] \chi_{n} \bigr\rangle
    &= (n_{j}+1) \frac{\hbar}{2m}\, \mathcal{N}_{\hbar}(\mathcal{B},\delta)\,\mathcal{B}^{-1/2}_{jk}
      (\mathcal{A} - \rmi\mathcal{B})_{kl} (\mathcal{A} + \rmi\mathcal{B})_{lr} \mathcal{B}^{-1/2}_{rj} \\
    &= (n_{j}+1) \frac{\hbar}{2m}\, \mathcal{N}_{\hbar}(\mathcal{B},\delta) 
      \parentheses{ \mathcal{B}^{-1/2}(\mathcal{A}^{2} + \mathcal{B}^{2})\mathcal{B}^{-1/2} }_{jj}
  \end{align*}
  since both $\mathcal{A}$ and $\mathcal{B}$ are symmetric.
  As a result, we obtain the recurrence relation
  \begin{equation*}
    \texval{\hat{T}}^{(n+e_{j})} = \texval{\hat{T}}^{(n)}
    + \frac{\hbar}{2m}\, \mathcal{N}_{\hbar}(\mathcal{B},\delta) 
    \parentheses{ \mathcal{B}^{-1/2}(\mathcal{A}^{2} + \mathcal{B}^{2})\mathcal{B}^{-1/2} }_{jj}.
  \end{equation*}
  It is easy to see by direct calculations that, as in \cite{OhLe2013},
  \begin{equation*}
    \texval{\hat{T}}^{(0)} = \mathcal{N}_{\hbar}(\mathcal{B},\delta) \braces{
      \frac{p^{2}}{2m} + \frac{\hbar}{4m} \tr\parentheses{ \mathcal{B}^{-1/2}(\mathcal{A}^{2} + \mathcal{B}^{2})\mathcal{B}^{-1/2} }
    }.
  \end{equation*}
  Hence the recurrence relation yields
  \begin{align*}
    \texval{\hat{T}}^{(n)}
    &= \texval{\hat{T}}^{(0)}
      + \frac{\hbar}{2m}\, \mathcal{N}_{\hbar}(\mathcal{B},\delta) \sum_{j=1}^{d} n_{j} \parentheses{ \mathcal{B}^{-1/2}(\mathcal{A}^{2} + \mathcal{B}^{2})\mathcal{B}^{-1/2} }_{jj} \\
    &= \mathcal{N}_{\hbar}(\mathcal{B},\delta) \braces{
      \frac{p^{2}}{2m}
      + \frac{\hbar}{4m} \tr\parentheses{ \mathcal{B}^{-1/2} \Lambda^{(n)} \mathcal{B}^{-1/2}(\mathcal{A}^{2} + \mathcal{B}^{2}) }
      } \\
    &= \mathcal{N}_{\hbar}(\mathcal{B},\delta) \braces{
      \frac{p^{2}}{2m}
      + \frac{\hbar}{4m} \tr\parentheses{ (\mathcal{B}^{(n)})^{-1}(\mathcal{A}^{2} + \mathcal{B}^{2}) }
      },
  \end{align*}
  where $\Lambda^{(n)}$ and $\mathcal{B}^{(n)}$ are defined in \eqref{eq:Lambda_n} and \eqref{eq:mathcalB_n}.
  
  Next, let us find an asymptotic expansion of the potential term $\exval{V}^{(n)}$.
  We mimic the technique employed in the proof of Theorem~2.9 in \citet{Ha1998}.
  First, since $V$ is assumed to be $C^{3}$, we have, for any $x \in \R^{d}$,
  \begin{equation*}
    V(x) = V(q) + D_{k}V(q) (x - q)_{k} + \frac{1}{2} D^{2}_{kl}V(q) (x - q)^{\otimes^{2}}_{kl}
    + \sum_{|\alpha|=3} \frac{1}{\alpha!} D^{\alpha}V(\sigma(q,x)) (x - q)^{\alpha}
  \end{equation*}
  for some point $\sigma(q,x)$ in the closed ball $\bar{\mathbb{B}}_{|x-q|}(q) \subset \R^{d}$ with radius $|x-q|$ centered at $q$, where we used the shorthands $(x - q)^{\otimes^{2}}_{kl} = (x - q)_{k} (x - q)_{l}$ and $(x - q)^{\alpha} = \prod_{j=1}^{d} (x - q)_{j}^{\alpha_{j}}$.
  Therefore,
  \begin{align*}
    \exval{V}^{(n)}
    &= \ip{ \chi_{n} }{ V \chi_{n} } \\
    &= \norm{ \chi_{n} }^{2}\, V(q)
      + \ip{ \chi_{n} }{ (x - q)_{k} \chi_{n} } D_{k}V(q) 
     + \frac{1}{2}\ip{ \chi_{n} }{ (x - q)^{\otimes^{2}}_{kl} \chi_{n} } D^{2}_{kl}V(q) \\
    &\quad + \sum_{|\alpha|=3} \frac{1}{\alpha!} \ip{ \chi_{n} }{ D^{\alpha}V(\sigma(q,x)) (x - q)^{\alpha} \chi_{n} }.
   \end{align*}
  But then the second term vanishes because, using \eqref{eq:x-q_in_As} and in view of \eqref{eq:chi_n-raised} and \eqref{eq:chi_n-lowered},
  \begin{align*}
    \ip{ \chi_{n} }{ (x - q)_{k} \chi_{n} }
    &= \sqrt{\frac{\hbar}{2}}\,\mathcal{B}^{-1/2}_{kl} \ip{ \chi_{n} }{ ( \mathscr{L}_{l} + \mathscr{L}^{*}_{l} ) \chi_{n} } \\
    &= \sqrt{\frac{\hbar}{2}}\,\mathcal{B}^{-1/2}_{kl}\parentheses{
      \sqrt{n_{l}} \ip{ \chi_{n} }{ \chi_{n-e_{l}} }
      + \sqrt{n_{l}+1} \ip{ \chi_{n} }{ \chi_{n+e_{l}} }
      } \\
    &= 0.
  \end{align*}
  On the other hand, using \eqref{eq:x-q_in_As} and \eqref{eq:inner_products} as well as $\ip{ \mathscr{L}_{j} \chi_{n} }{ \mathscr{L}_{k} \chi_{n} } = \delta_{jk}\, n_{j} \mathcal{N}_{\hbar}(\mathcal{B},\delta)$, we can evaluate the third term as follows:
  \begin{align*}
    \ip{ \chi_{n} }{ (x - q)^{\otimes^{2}}_{kl} \chi_{n} } 
    &= \ip{ (x - q)_{k} \chi_{n} }{ (x - q)_{l} \chi_{n} } \\
    &= \frac{\hbar}{2}\,\mathcal{B}^{-1/2}_{kr}\mathcal{B}^{-1/2}_{ls}
      \bigip{ (\mathscr{L}_{r} + \mathscr{L}_{r}^{*})\chi_{n} }{ (\mathscr{L}_{s} + \mathscr{L}_{s}^{*}) \chi_{n} } \\
    &= \frac{\hbar}{2}\,\mathcal{B}^{-1/2}_{kr}\mathcal{B}^{-1/2}_{ls} \parentheses{
      \bigip{ \mathscr{L}_{r} \chi_{n} }{ \mathscr{L}_{s} \chi_{n} }
      + \bigip{ \mathscr{L}^{*}_{r} \chi_{n} }{ \mathscr{L}^{*}_{s} \chi_{n} }
      } \\
    &= \frac{\hbar}{2}\,\mathcal{N}_{\hbar}(\mathcal{B},\delta)\,\mathcal{B}^{-1/2}_{kr}(2n_{r}+1)\delta_{rs}\mathcal{B}^{-1/2}_{ls} \\
    &= \frac{\hbar}{2}\,\mathcal{N}_{\hbar}(\mathcal{B},\delta)\,\parentheses{ \mathcal{B}^{-1/2}\Lambda^{(n)}\mathcal{B}^{-1/2} }_{kl} \\
    &= \frac{\hbar}{2}\,\mathcal{N}_{\hbar}(\mathcal{B},\delta)\, (\mathcal{B}^{(n)})^{-1}_{kl},
  \end{align*}
  where $\Lambda^{(n)}$ and $\mathcal{B}^{(n)}$ are defined in \eqref{eq:Lambda_n} and \eqref{eq:mathcalB_n}.
  So it remains to show that the last term is $O(\hbar^{3/2})$.
  Let $R > 0$ and set
  \begin{equation*}
    C_{3} \defeq \max_{|\alpha|=3} \max_{x \in \bar{\mathbb{B}}_{R}(q)} \abs{ D^{\alpha}V(x) }.
  \end{equation*}
  If $x \in \bar{\mathbb{B}}_{R}(q)$ then $\sigma(q,x) \in \bar{\mathbb{B}}_{R}(q)$ as well and hence, for any $\alpha \in \N_{0}^{d}$ with $|\alpha| = 3$,
  \begin{align*}
    \abs{ D^{\alpha}V(\sigma(q,x)) (x - q)^{\alpha} \chi_{n}(y;x) }
    \le C_{3}\,\hbar^{3/2}\, \abs{ \parentheses{ \frac{x - q}{\sqrt{\hbar}} }^{\alpha} \chi_{n}(y;x) },
  \end{align*}
  whereas if $x \in \bar{\mathbb{B}}_{R}(q)^{\rm c}$ then, due to the assumption~\eqref{eq:assumption_on_D3V} on the potential $V$, there exists $C_{4} > 0$ such that
  \begin{align*}
    \abs{ D^{\alpha}V(\sigma(q,x)) (x - q)^{\alpha} \chi_{n}(y;x) }
    &\le C_{4} \exp(M_{1}|x|^{2}) \abs{ \chi_{n}(y;x) }.
  \end{align*}
  Let $\mathbf{1}_{S}$ be the indicator function of an arbitrary subset $S \subset \R^{d}$ and also define the normalized wave packet
  \begin{equation*}
    \varphi_{n}(y;x) \defeq \frac{\chi_{n}(y;x)}{\norm{\chi_{n}(y;\,\cdot\,)}}.
  \end{equation*}
  Then we have, for any $\alpha \in \N_{0}^{d}$ with $|\alpha| = 3$ and any $x \in \R^{d}$,
  \begin{align*}
    \abs{ D^{\alpha}V(\sigma(q,x)) (x - q)^{\alpha} \chi_{n}(y;x) }
    &= \mathbf{1}_{\bar{\mathbb{B}}_{R}(q)}(x) \abs{ D^{\alpha}V(\sigma(q,x)) (x - q)^{\alpha} \chi_{n}(y;x) } \\
    &\qquad + \mathbf{1}_{\bar{\mathbb{B}}_{R}(q)^{\rm c}}(x) \abs{ D^{\alpha}V(\sigma(q,x)) (x - q)^{\alpha} \chi_{n}(y;x) } \\
    &\le C_{3}\, \hbar^{3/2}\, \norm{\chi_{n}(y;\,\cdot\,)}\, \mathbf{1}_{\bar{\mathbb{B}}_{R}(q)}(x)\,\abs{ \parentheses{ \frac{x - q}{\sqrt{\hbar}} }^{\alpha} \varphi_{n}(y;x) } \\
    &\quad+ C_{4} \norm{\chi_{n}(y;\,\cdot\,)}\, \mathbf{1}_{\bar{\mathbb{B}}_{R}(q)^{\rm c}}(x)\, \exp(M_{1}|x|^{2}) \abs{ \varphi_{n}(y;x) }.
  \end{align*}
  However, the norm of the first term is $O(\hbar^{3/2})$ because, as shown in \citet[Eq.~(3.30)]{Ha1998},
  \begin{equation*}
    \norm{ \parentheses{ \frac{x - q}{\sqrt{\hbar}} }^{\alpha} \varphi_{n}(y;\,\cdot\,) } = O(1),
  \end{equation*}
  whereas
  \begin{equation*}
    \norm{ \exp(M_{1}|x|^{2}) \varphi_{n}(y;\,\cdot\,) } = o(\hbar^{\gamma})
  \end{equation*}
  for any real number $\gamma$.
  Hence
  \begin{equation*}
    \norm{ D^{\alpha}V(\sigma(q,x)) (x - q)^{\alpha} \chi_{n} } = \norm{\chi_{n}(y;\,\cdot\,)}\,O(\hbar^{3/2}),
  \end{equation*}
  and thus by the Cauchy--Schwarz inequality,
  \begin{align*}
    \sum_{|\alpha|=3} \frac{1}{\alpha!} \abs{ \ip{\chi_{n}}{ D^{\alpha}V(\sigma(q,x)) (x - q)^{\alpha} \chi_{n} } }
    &\le \norm{\chi_{n}} \sum_{|\alpha|=3} \frac{1}{\alpha!} \norm{ D^{\alpha}V(\sigma(q,x)) (x - q)^{\alpha} \chi_{n} } \\
    &\le \norm{\chi_{n}}^{2}\,O(\hbar^{3/2}) \\
    &= \mathcal{N}_{\hbar}(\mathcal{B},\delta)\,O(\hbar^{3/2}).
  \end{align*}
  As a result, we obtain
  \begin{align*}
    \exval{V}^{(n)}
    &= \mathcal{N}_{\hbar}(\mathcal{B},\delta) \parentheses{
      V(q) + \frac{\hbar}{4}\,(\mathcal{B}^{(n)})^{-1}_{kl} D^{2}_{kl}V(q)
      + O(\hbar^{3/2})
      } \\
    &= \mathcal{N}_{\hbar}(\mathcal{B},\delta) \parentheses{
      V(q) + \frac{\hbar}{4}\tr\!\parentheses{ (\mathcal{B}^{(n)})^{-1} D^{2}V(q) }
      + O(\hbar^{3/2})
      }.
  \end{align*}
  Hence we have the asymptotic expansion~\eqref{eq:exvalH-asymptotic} along with \eqref{eq:H_M}.
\end{proof}

\subsection{Hamiltonian Dynamics of Semiclassical Wave Packets}
Now that we have both the symplectic forms $\{ \Omega_{\mathcal{M}}^{(n)} \}_{n\in\N_{0}^{d}}$ and the Hamiltonians $\{ H^{(n)} \}_{n\in\N_{0}^{d}}$ associated with the semiclassical wave packets $\{ \chi_{n}(y;\,\cdot\,) \}_{n\in\N_{0}^{d}}$, we may formulate Hamiltonian dynamics for each of them:
\begin{theorem}[Hamiltonian dynamics of semiclassical wave packets]
  \label{thm:Hamiltonian_system}
  Suppose that the potential $V$ satisfies the conditions stated in Proposition~\ref{prop:H_M}.
  Then, for any $n \in \N_{0}^{d}$, the Hamiltonian vector field $X_{H^{(n)}} \in \mathfrak{X}(\mathcal{M})$ associated with the semiclassical wave packet $\chi_{n}(y;\,\cdot\,)$ is defined by ${\bf i}_{X_{H^{(n)}}} \Omega_{\mathcal{M}}^{(n)} = \mathbf{d}H^{(n)}$ or
  \begin{equation}
    \label{eq:SemiclassicalSystem}
    \begin{array}{c}
      \DS
      \dot{q} = \frac{p}{m},
      \qquad
      \dot{p} = -D_{q}V^{(n)}_{\hbar}(q,\mathcal{B}),
      \medskip\\
      \DS
      \dot{\mathcal{A}} = -\frac{1}{m}\parentheses{ \mathcal{A}^{2} - \frac{1}{2}(\mathcal{B}^{(n)}\Lambda^{(n)}\mathcal{B} + \mathcal{B}\Lambda^{(n)}\mathcal{B}^{(n)}) } - D^{2}V(q),
      \qquad
      \dot{\mathcal{B}}^{(n)} = -\frac{1}{m}(\mathcal{A}\mathcal{B}^{(n)} + \mathcal{B}^{(n)}\mathcal{A}),
      \medskip\\
      \DS
      \dot{\phi} = \frac{p^{2}}{2m} - V(q) - \frac{\hbar}{2m} \tr(\Lambda^{(n)}\mathcal{B}),
      \qquad
      \dot{\delta} = \frac{\hbar}{2m} \tr\mathcal{A},
    \end{array}
  \end{equation}
  where the corrected potential $V^{(n)}_{\hbar}$ is defined as
  \begin{equation}
    \label{eq:V^n}
    V^{(n)}_{\hbar}(q,\mathcal{B})
    \defeq
    V(q) + \frac{\hbar}{4}\tr\!\parentheses{ (\mathcal{B}^{(n)})^{-1} D^{2}V(q) }.
  \end{equation}
\end{theorem}

\begin{remark}
  For the special case with $n = (\bar{n}, \dots, \bar{n}) \in \N_{0}^{d}$ with $\bar{n} \in \N_{0}$, we have $\Lambda^{(n)} = (2\bar{n}+1)\,I_{d}$ and $\mathcal{B}^{(n)} = (2\bar{n}+1)^{-1} \mathcal{B}$.
  Hence the equations for $\mathcal{A}$ and $\mathcal{B}$ simplify to those in \eqref{eq:Heller}.
\end{remark}

\begin{proof}
  The assertion follows from tedious but straightforward calculations using the formulas for the symplectic forms $\{ \Omega_{\mathcal{M}}^{(n)} \}_{n\in\N_{0}^{d}}$ and the Hamiltonians $\{ H^{(n)} \}_{n\in\N_{0}^{d}}$ from Propositions~\ref{prop:Omega_M} and \ref{prop:H_M}, respectively.
\end{proof}

\section{Symplectic Reduction by Phase Symmetry}
\label{sec:Symplectic_Reduction}
The Hamiltonian $H^{(n)}$ found in \eqref{eq:H_M} does not depend on the phase variable $\phi$ and hence is invariant under the $\mathbb{S}^{1}$ phase shift action.
Therefore, we can reduce the Hamiltonian dynamics~\eqref{eq:SemiclassicalSystem} to a lower-dimensional one by the symplectic (Marsden--Weinstein) reduction~\cite{MaWe1974} (see also \citet[Sections~1.1 and 1.2]{MaMiOrPeRa2007}) as is done for the Gaussian case in our earlier work~\cite{OhLe2013}.
The resulting {\em reduced} symplectic structure is much simpler than $\Omega_{\mathcal{M}}^{(n)}$ from Proposition~\ref{prop:H_M} and moreover takes an appealing form: It is given by the classical symplectic form plus an $O(\hbar)$ ``correction'' term for any $n \in \N_{0}^{d}$ and similarly for the Hamiltonian as well, hence generalizing the results for the Gaussian case ($n = 0$) from our earlier work~\cite[Theorem~4.1]{OhLe2013}.

\subsection{Reduced Symplectic Structures}
\begin{proposition}[Reduced semiclassical symplectic structures]
  \label{prop:Omega_M-reduced}
  Let $\Phi\colon \mathbb{S}^{1} \times \mathcal{M} \to \mathcal{M}$ be the $\mathbb{S}^{1}$-action on $\mathcal{M}$ defined for any $\theta \in \mathbb{S}^{1}$ as
  \begin{equation}
    \label{eq:Phi}
    \Phi_{\theta}: \mathcal{M} \to \mathcal{M};
    \quad
    (q, p, \mathcal{A}, \mathcal{B}, \phi, \delta) \mapsto (q, p, \mathcal{A}, \mathcal{B}, \phi + \hbar\,\theta, \delta).
  \end{equation}
  Then the corresponding momentum map $\mathbf{J}_{\!\mathcal{M}}^{(n)}\colon \mathcal{M} \to \mathfrak{so}(2)^{*} \cong \R$ is given by
  \begin{equation}
    \label{eq:mathbfJ}
    \mathbf{J}_{\!\mathcal{M}}^{(n)}(y) = -\hbar\,\mathcal{N}_{\hbar}(\mathcal{B},\delta),
  \end{equation}
  and the Marsden--Weinstein quotient
  \begin{equation*}
    \overline{\mathcal{M}}_{\hbar}^{(n)} \defeq (\mathbf{J}_{\!\mathcal{M}}^{(n)})^{-1}(-\hbar)/\mathbb{S}^{1} = T^{*}\R^{d} \times \Sigma_{d}
  \end{equation*}
  is equipped with the reduced symplectic form
  \begin{equation}
    \label{eq:Omega-reduced}
    \begin{split}
      \overline{\Omega}_{\hbar}^{(n)}
      &= \mathbf{d}q_{i} \wedge \mathbf{d}p_{i} + \frac{\hbar}{4} (\mathcal{B}^{(n)})^{-1}_{jr} (\mathcal{B}^{(n)})^{-1}_{sk}\,\d\mathcal{A}_{jk} \wedge \d\mathcal{B}^{(n)}_{rs}
      \\
      &= \mathbf{d}q_{i} \wedge \mathbf{d}p_{i} + \frac{\hbar}{4} \d(\mathcal{B}^{(n)})^{-1}_{jk} \wedge \d\mathcal{A}_{jk},
    \end{split}
  \end{equation}
  where $\mathcal{B}^{(n)}$ is defined in \eqref{eq:mathcalB_n}.
\end{proposition}

\begin{proof}
  Let us first find the momentum map corresponding to the above action.
  It is easy to see that the action $\Phi$ leaves the one-form $\Theta_{\mathcal{M}}^{(n)}$ invariant, i.e., $\Phi_{\theta}^{*}\Theta_{\mathcal{M}}^{(n)} = \Theta_{\mathcal{M}}^{(n)}$ for any $\theta \in \mathbb{S}^{1}$, and hence $\Phi$ is symplectic with respect to $\Omega_{\mathcal{M}}^{(n)}$, i.e., $\Phi_{\theta}^{*}\Omega_{\mathcal{M}}^{(n)} = \Omega_{\mathcal{M}}^{(n)}$ for any $\theta \in \mathbb{S}^{1}$.
  The infinitesimal generator of the above action corresponding to an arbitrary element $\xi$ in the Lie algebra $\mathfrak{so}(2) \cong \R$ is 
  \begin{equation*}
    \xi_{\mathcal{M}}(y) \defeq \left. \od{}{\eps} \Phi_{\eps\xi}(y) \right|_{\eps=0}
    = \hbar\,\xi\,\pd{}{\phi}.
  \end{equation*}
  Since $\Phi$ leaves $\Theta_{\mathcal{M}}^{(n)}$ invariant for any $n \in \N_{0}^{d}$, the momentum map $\mathbf{J}_{\!\mathcal{M}}^{(n)}$ with respect to the symplectic structure $\Omega_{\mathcal{M}}^{(n)}$ for any $n \in \N_{0}^{d}$ is defined as
  \begin{equation*}
    \ip{ \mathbf{J}_{\!\mathcal{M}}^{(n)}(y) }{\xi} = \ip{ \Theta_{\mathcal{M}}^{(n)}(y) }{ \xi_{\mathcal{M}}(y) } = -\hbar\,\mathcal{N}_{\hbar}(\mathcal{B},\delta)\,\xi
  \end{equation*}
  for any $\xi \in \mathfrak{so}(2)$.
  Hence we obtain \eqref{eq:mathbfJ}.
  So the level set $(\mathbf{J}_{\!\mathcal{M}}^{(n)})^{-1}(-\hbar)$ is given by the set of those parameters that normalize the wave packet $\chi_{n}(y;\,\cdot\,)$, i.e.,
  \begin{align*}
    (\mathbf{J}_{\!\mathcal{M}}^{(n)})^{-1}(-\hbar)
    &= \setdef{ (q,p,\mathcal{A},\mathcal{B},\phi,\delta) \in \mathcal{M} }{ \mathcal{N}_{\hbar}(\mathcal{B},\delta) = 1 } \\
    &= \setdef{ y = (q,p,\mathcal{A},\mathcal{B},\phi,\delta) \in \mathcal{M} }{ \bigl\| \chi_{n}(y;\,\cdot\,) \bigr\| = 1 }.
  \end{align*}
  Therefore, one may solve $\mathcal{N}_{\hbar}(\mathcal{B},\delta) = 1$ for $\delta$ (see \eqref{eq:N}) to have the inclusion $i_{\hbar}\colon (\mathbf{J}_{\!\mathcal{M}}^{(n)})^{-1}(-\hbar) \to \mathcal{M}$ defined as
  \begin{equation*}
    i_{\hbar}\colon (q,p,\mathcal{A},\mathcal{B},\phi) \mapsto \parentheses{ q,p,\mathcal{A},\mathcal{B},\phi,\frac{\hbar}{4}\ln\parentheses{\frac{(\pi\hbar)^{d}}{\det\mathcal{B}}} },
  \end{equation*}
  and hence we have
  \begin{equation*}
    i_{\hbar}^{*}\Theta_{\mathcal{M}}^{(n)}
    = p_{i}\,\mathbf{d}q_{i} - \frac{\hbar}{4}\tr\parentheses{ (\mathcal{B}^{(n)})^{-1} \d\mathcal{A} } - \mathbf{d}\phi.
  \end{equation*}
  Therefore,
  \begin{align*}
    i_{\hbar}^{*}\Omega_{\mathcal{M}}^{(n)}
    &= -i_{\hbar}^{*}\mathbf{d}\Theta_{\mathcal{M}}^{(n)} \\
    &= -\mathbf{d}\bigl( i_{\hbar}^{*}\Theta_{\mathcal{M}}^{(n)} \bigr) \\
    &= \mathbf{d}q_{i} \wedge \mathbf{d}p_{i} + \frac{\hbar}{4}\,\d(\mathcal{B}^{(n)})^{-1}_{jk} \wedge \d\mathcal{A}_{jk}.
  \end{align*}
  However, we have the quotient map
  \begin{equation*}
    \pi_{\hbar}\colon (\mathbf{J}_{\!\mathcal{M}}^{(n)})^{-1}(-\hbar) \to \overline{\mathcal{M}}_{\hbar}^{(n)} \defeq (\mathbf{J}_{\!\mathcal{M}}^{(n)})^{-1}(-\hbar)/\mathbb{S}^{1};
    \quad
    (q,p,\mathcal{A},\mathcal{B},\phi) \mapsto (q,p,\mathcal{A},\mathcal{B}),
  \end{equation*}
  and see that $\overline{\Omega}_{\hbar}^{(n)}$ shown in \eqref{eq:Omega-reduced} satisfies $\pi_{\hbar}^{*}\overline{\Omega}_{\hbar}^{(n)} = i_{\hbar}^{*}\Omega_{\mathcal{M}}^{(n)}$ and hence defines the reduced symplectic form on $\overline{\mathcal{M}}_{\hbar}^{(n)}$; note that $\pi_{\hbar}^{*}$ is injective because $\pi_{\hbar}$ is a surjective submersion.
\end{proof}

\subsection{Reduced Hamiltonian Dynamics}
Since the Hamiltonian $H^{(n)}$ does not depend on the phase variable $\phi$, it has the $\mathbb{S}^{1}$-symmetry under the action $\Phi$ defined in \eqref{eq:Phi}.
Therefore we can reduce the Hamiltonian dynamics~\eqref{eq:SemiclassicalSystem} to the reduced symplectic manifold $\overline{\mathcal{M}}_{\hbar}$:
\begin{theorem}[Reduced semiclassical wave packet dynamics]
  \label{thm:ReducedSemiclassicalSystem}
  Suppose that the potential $V$ satisfies the conditions stated in Proposition~\ref{prop:H_M}.
  Then the Hamiltonian system~\eqref{eq:SemiclassicalSystem} on $\mathcal{M}$ for the semiclassical wave packet $\chi_{n}(y;\,\cdot\,)$ is reduced by the above $\mathbb{S}^{1}$-symmetry to the Hamiltonian system
  \begin{equation}
    \label{eq:ReducedHamiltonianSystem}
    {\bf i}_{X_{\overline{H}^{(n)}_{\hbar}}} \overline{\Omega}^{(n)}_{\hbar} = \mathbf{d}\overline{H}^{(n)}_{\hbar}
  \end{equation}
  defined on $\overline{\mathcal{M}}_{\hbar}^{(n)} = T^{*}\R^{d} \times \Sigma_{d}$ with the reduced symplectic form~\eqref{eq:Omega-reduced} and the reduced Hamiltonian $\overline{H}^{(n)}\colon \overline{\mathcal{M}}_{\hbar}^{(n)} \to \R$ given by
  \begin{equation}
    \label{eq:H-reduced}
    \begin{split}
      \overline{H}^{(n)}_{\hbar}
      &= \frac{p^{2}}{2m} + V(q)
      + \frac{\hbar}{4} \parentheses{
        \frac{1}{m} \tr\!\parentheses{ (\mathcal{B}^{(n)})^{-1}(\mathcal{A}^{2} + \mathcal{B}^{2}) }
        + \tr\parentheses{ (\mathcal{B}^{(n)})^{-1} D^{2}V(q) }
      } \\
      &= \frac{p^{2}}{2m} + \frac{\hbar}{4m} \tr\!\parentheses{ (\mathcal{B}^{(n)})^{-1}(\mathcal{A}^{2} + \mathcal{B}^{2}) } + V^{(n)}_{\hbar}(q,\mathcal{B}),
    \end{split}
  \end{equation}
  where the corrected potential $V^{(n)}_{\hbar}$ is defined in \eqref{eq:V^n}.
  Specifically, \eqref{eq:ReducedHamiltonianSystem} gives the reduced set of the semiclassical equations:
  \begin{equation}
    \label{eq:ReducedSemiclassicalSystem}    
    \begin{array}{c}
      \DS
      \dot{q} = \frac{p}{m},
      \qquad
        \dot{p} = -D_{q}V^{(n)}_{\hbar}(q,\mathcal{B}),
      \medskip\\
      \DS
      \dot{\mathcal{A}} = -\frac{1}{m}\parentheses{ \mathcal{A}^{2} - \frac{1}{2}(\mathcal{B}^{(n)}\Lambda^{(n)}\mathcal{B} + \mathcal{B}\Lambda^{(n)}\mathcal{B}^{(n)}) } - D^{2}V(q),
      \qquad
      \dot{\mathcal{B}}^{(n)} = -\frac{1}{m}(\mathcal{A}\mathcal{B}^{(n)} + \mathcal{B}^{(n)}\mathcal{A}).
    \end{array}
  \end{equation}
\end{theorem}

\begin{proof}
  Clearly the Hamiltonian $H^{(n)}$ has the $\mathbb{S}^{1}$-symmetry, i.e., $H^{(n)} \circ \Phi_{\theta} = H^{(n)}$ for any $\theta \in \mathbb{S}^{1}$ and so, by the Marsden--Weinstein reduction~\cite{MaWe1974} (see also \citet[Sections~1.1 and 1.2]{MaMiOrPeRa2007}), the Hamiltonian system~\eqref{eq:SemiclassicalSystem} reduces to the reduced one \eqref{eq:ReducedHamiltonianSystem} with the reduced Hamiltonian $\overline{H}^{(n)}_{\hbar}$ defined as $\overline{H}^{(n)}_{\hbar} \circ \pi_{\hbar} = H^{(n)} \circ i_{\hbar}$, which yields~\eqref{eq:H-reduced}.
  It is a straightforward calculation to see that \eqref{eq:ReducedHamiltonianSystem} yields \eqref{eq:ReducedSemiclassicalSystem}.
\end{proof}

\section{Numerical Results}
\label{sec:Numerical_Results}
\subsection{Problem Setting: Escape from Cubic Potential Well}
We performed numerical experiments with the simple one-dimensional potential (i.e., $d = 1$)
\begin{equation}
  \label{eq:V}
  V(x) = 2x^{2} + x^{3} + 0.1x^{4}
\end{equation}
and $m = 1$, and different values of index $n$ and parameter $\hbar$.
This example is a slightly modified version of an example from \citet[Section~6.4]{KeLaOh2016}, which in turn is a rescaled version of the cubic potential example from \citet{PrPe2000} with an additional quartic confinement term to make sure that the potential is bounded from below; see the assumptions in Proposition~\ref{prop:H_M}.

The initial position of the particle in the phase space is $(q(0),p(0)) = (0.25,1)$; this gives the classical total energy $H_{\text{cl}} \simeq 0.641$.
This is below the local maximum $V_{1} \simeq 1.703$ of the potential (located at $x \simeq 1.73$) and hence the solution $(q(t),p(t))$ to the classical Hamiltonian system in \eqref{eq:Heller} gives a periodic orbit confined in the potential well; see Fig.~\ref{fig:V}.
\begin{figure}[htbp]
  \begin{center}
    \begin{tikzpicture}[smooth, domain=-3:2, scale=2, samples=50]
      \clip (-3,-0.35) rectangle (1,2);
      \draw[-] (-3,0) -- (0.75,0) node[below] {$x$};
      \draw[-] (0,-1) -- (0,2.5);
      \draw[thick] plot (\x,{2*pow(\x,2) + pow(\x,3) + 0.1*pow(\x,4)});
      \node[below left] at (0.8,2) {$V(x)$};
      \draw[-, thin] (0.25,0.075) -- (0.25,-0.075) node[below] {$-0.25$};
      \fill [green!55!black] (0.25,0) circle (1pt);
      \draw[-stealth, ultra thick, green!55!black] (0.25,0) -- (-0.15,0);
      \node[below] at (-10/3,0) {$-10/3$};
      \draw[dashed] (-10/3,0) -- (-10/3,1.85);
      \draw[-, thick, blue!55!black] (-0.686347,0.641016) -- (0.503472,0.641016);
      \node[above left, blue!55!black] at (0,0.641016) {$H_{\text{cl}}$};
      \draw[dashed] (-1.73444,1.70386) -- (0,1.70386);
      \node[above left] at (0,1.70386) {$V_{1} \simeq 1.703$};
      \draw[dashed] (-1.73444,1.70386) -- (-1.73444,0) node[below] {$-1.73$};
      \draw[-, thin] (-2.5,0.075) -- (-2.5,-0.075) node[below] {$-2.5$};
    \end{tikzpicture}
  \end{center}
  \captionsetup{width=0.9\textwidth}
  \caption{
    Potential \eqref{eq:V}. $H_{\text{cl}}$ is the classical Hamiltonian $p^{2}/(2m) + V(q)$ with initial condition $(q(0),p(0)) = (0.25,1)$; the green dot and arrow on the $x$-axis indicate the initial position and velocity of the particle.
    The fact that $H_{\text{cl}} < V_{1}$ implies that the classical trajectory following \eqref{eq:Heller} is trapped in the potential well.
  }
  \label{fig:V}
\end{figure}
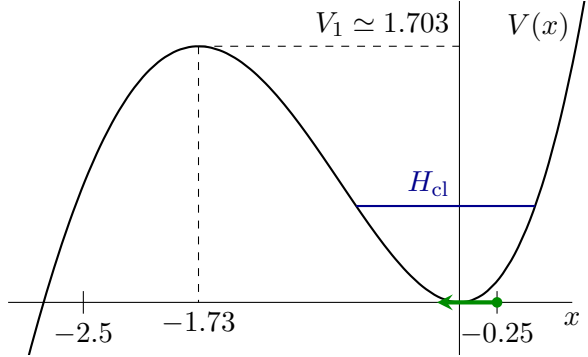

However, the semiclassical Hamiltonian dynamics~\eqref{eq:ReducedSemiclassicalSystem} with the same initial condition $(q(0),p(0))$ may not be confined in the potential well because it is a Hamiltonian system in the higher-dimensional space $T^{*}\R^{d} \times \Sigma_{d}$.
We set the initial condition for $(\mathcal{A},\mathcal{B})$ as $(\mathcal{A}(0),\mathcal{B}(0)) = (0,1)$; the phase is set as $\phi(0) = 0$ because it is irrelevant to the dynamics of observables; $\delta(0)$ is set so that the initial wave function is normalized.
This means that the initial wave function
\begin{equation}
  \label{eq:psi-IC}
  \psi(0,x) = \chi_{n}(q(0),p(0),\mathcal{A}(0),\mathcal{B}(0),\phi(0),\delta(0);x)
\end{equation}
is the Hermite function with index $n\in\N_{0}$ because the ground state $\chi_{0}$ and the ladder operators become those of the harmonic oscillator; see~\eqref{eq:chi_0} and \eqref{eq:ladder_operators}.

\subsection{Results}
We computed the solutions of the classical system (the first two equations of \eqref{eq:Heller}) and the Hamiltonian system~\eqref{eq:ReducedSemiclassicalSystem} for the semiclassical wave packet $\chi_{n}$.
Also, for a reference solution $(\exval{\hat{x}}(t),\exval{\hat{p}}(t))$ of the expectation value dynamics, we used a method based on Egorov's Theorem with the phase space density developed in \cite{KeLaOh2016} (essentially equivalent to the Initial Value Representation (IVR) method~\cite{Mi1970,Mi1974b,WaSuMi1998,Mi2001} often used by chemical physicists).
It is known that such Egorov-type/IVR algorithms give an $O(\hbar^{2})$ approximation to the expectation value dynamics of the Schr\"odinger equation~\eqref{eq:SchroedingerEq-coordinates}, and hence provide a very good alternative to the exact solution in the semiclassical regime $\hbar \ll 1$.

We used the St\"ormer--Verlet method~\cite{Ve1967} to solve the classical Hamiltonian system and the variational splitting integrator of \citet{FaLu2006} (see also \citet[Section~IV.4]{Lu2008}) for the semiclassical solution; the time step is $0.01$ in all the cases.
It is easy to show that the variational splitting integrator preserves the symplectic structure~\eqref{eq:Omega-reduced}, and its limit as $\hbar \to 0$ gives the St\"ormer--Verlet method~\cite{Ve1967}.
The Egorov-type algorithm involves averaging of solutions of the classical Hamiltonian system with numerous initial conditions sampled with respect to the initial phase space densities corresponding to the initial wave function~\eqref{eq:psi-IC}.
Again the classical Hamiltonian system is solved using the St\"ormer--Verlet method and $100,000$ initial conditions are sampled to ensure accuracy.

The phase space plots of the results with $n = 1,3,5,10$ are shown in Figs.~\ref{fig:0.05}--\ref{fig:0.01} for $\hbar = 0.05, 0.025, 0.01$.
Fig.~\ref{fig:0.025_t-H} shows the time evolution of the classical energy $H_{\text{cl}} = p^{2}/(2m) + V$ along the classical solution, the semiclassical energy $\overline{H}^{(n)}_{\hbar}$ along the semiclassical solution, and the expectation value $\texval{\hat{H}}$ along the solutions of the Egorov-type algorithm.
They are shown for $0 \le t \le T$ where $T \simeq 3.39$ is the period of the classical solution.
The classical solution is trapped inside the potential well as explained earlier.
On the other hand, the semiclassical energy or Hamiltonian $\overline{H}^{(n)}_{\hbar}$ in \eqref{eq:H-reduced} becomes larger for larger values of $\hbar$ and $n$, and significantly deviates from the classical energy $H_{\text{cl}}$; see Fig.~\ref{fig:0.025_t-H} to see how $\overline{H}^{(n)}_{\hbar}$ changes as $\hbar$ becomes larger.
As a result, the solutions escape from the potential well for relatively large values of $\hbar$ and $n$ whereas it is trapped inside the well for some small values of $\hbar$ and $n$.
However, note that, unlike the classical case, $\overline{H}^{(n)}_{\hbar} < V_{1}$ does not necessarily imply that the trajectory is trapped inside the well because, as mentioned above, the semiclassical dynamics is a Hamiltonian system on $T^{*}\R^{d} \times \Sigma_{d}$ and so the level set of the Hamiltonian does not necessarily define a closed curve in $T^{*}\R^{d}$ even with $\overline{H}^{(n)}_{\hbar} < V_{1}$.
This is in fact the case for, e.g., $\hbar = 0.025$ and $n = 5, 10$; see Figs.~\ref{fig:0.025} and \ref{fig:0.025_t-H}.

More importantly, {\em the semiclassical solutions show a very good agreement with the reference solutions computed by the Egorov-type algorithm}.
Note that these solutions are computed using completely different methods: one from a single semiclassical Hamiltonian system~\eqref{eq:ReducedSemiclassicalSystem} whereas the other by sampling numerous solutions of the classical Hamiltonian system.

However, there is an issue with the semiclassical solutions as well.
The solutions of the semiclassical system~\eqref{eq:ReducedSemiclassicalSystem} deviate from the reference solutions after a while as we can see in some of the solutions in the figures.
In fact, it is known that approximation methods using semiclassical wave packets are usually valid only in the Ehrenfest time scale, i.e., $t \sim \ln(1/\hbar)$; see, e.g., \citet{HaJo2000}, \citet{CoRo2012}, and \citet{ScVaTo2012}.
It is because the wave packet becomes very widespread, i.e., the parameter $\mathcal{B}$---which controls the width of the wave packet---becomes significantly small in the Ehrenfest time scale.
The Ehrenfest time scales in our settings are roughly the same as the period $T \simeq 3.39$ of the classical solution.
This issue seems to exacerbate the errors in the numerical solution; see, for example the behavior of the Hamiltonian for $\hbar = 0.025$ and $n = 10$ in Fig.~\ref{fig:0.025_t-H}.

\begin{figure}[htbp]
  \centering
  \subfigure[$n = 1$]{
    \includegraphics[width=.465\linewidth]{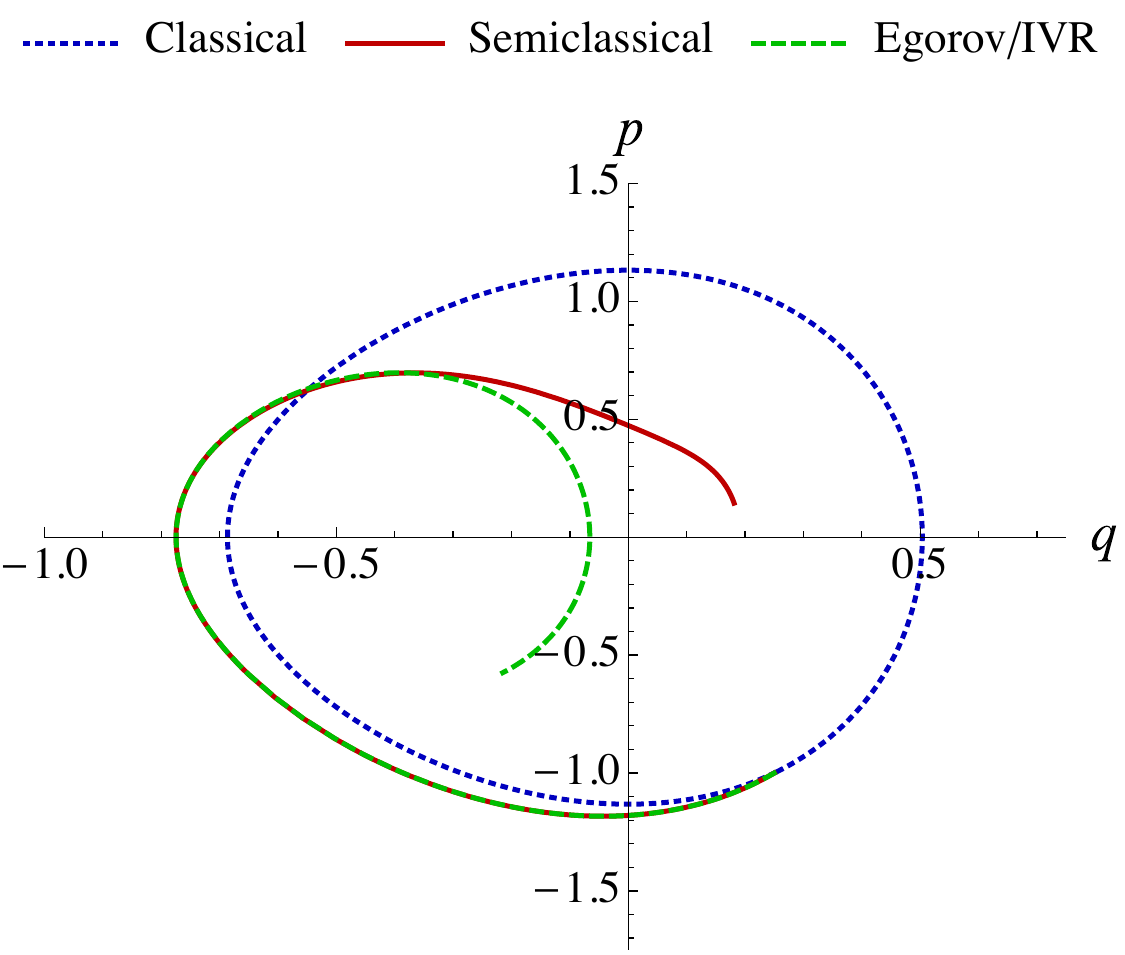}
  }
  \quad
  \subfigure[$n = 3$]{
    \includegraphics[width=.465\linewidth]{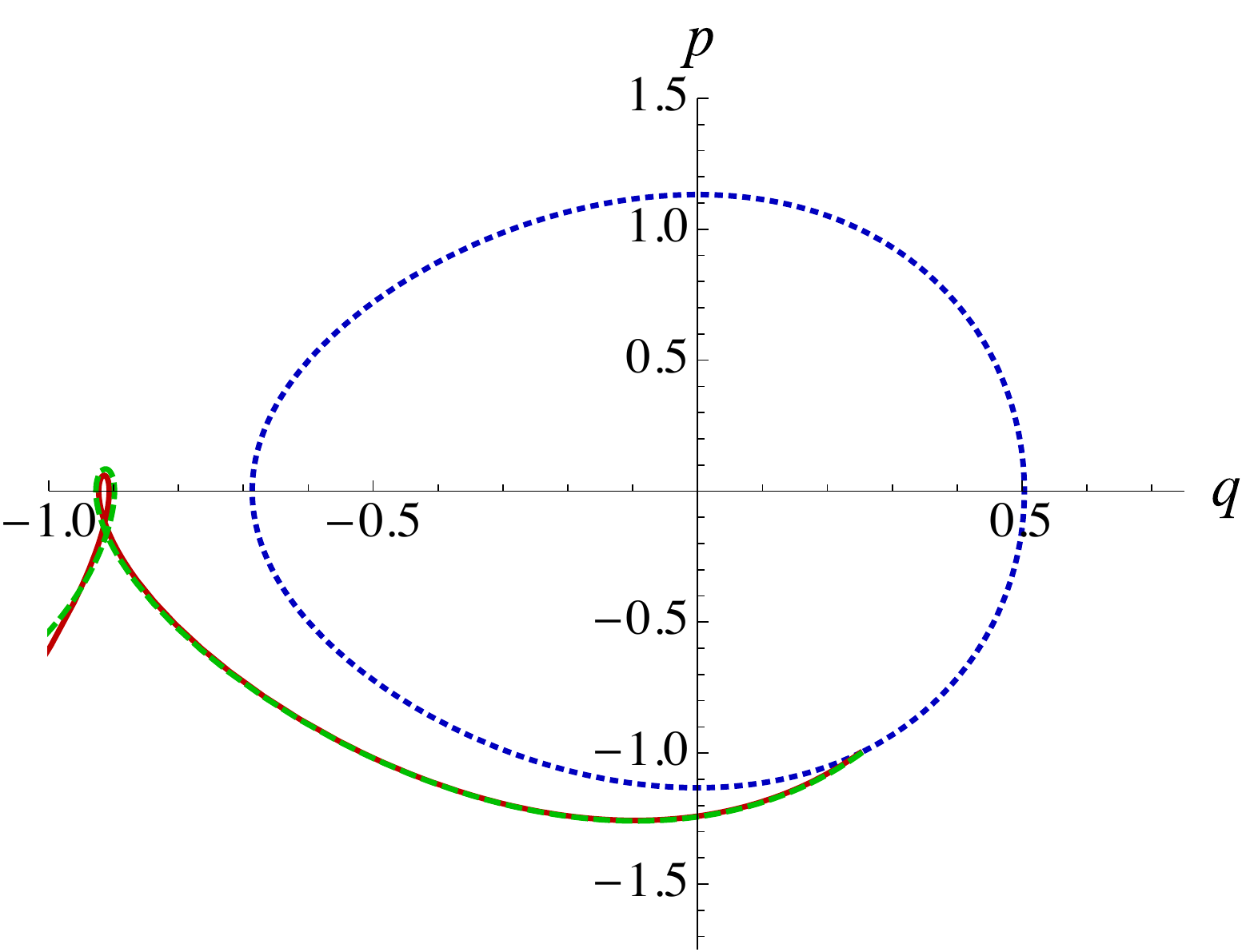}
  }
  \subfigure[$n = 5$]{
    \includegraphics[width=.465\linewidth]{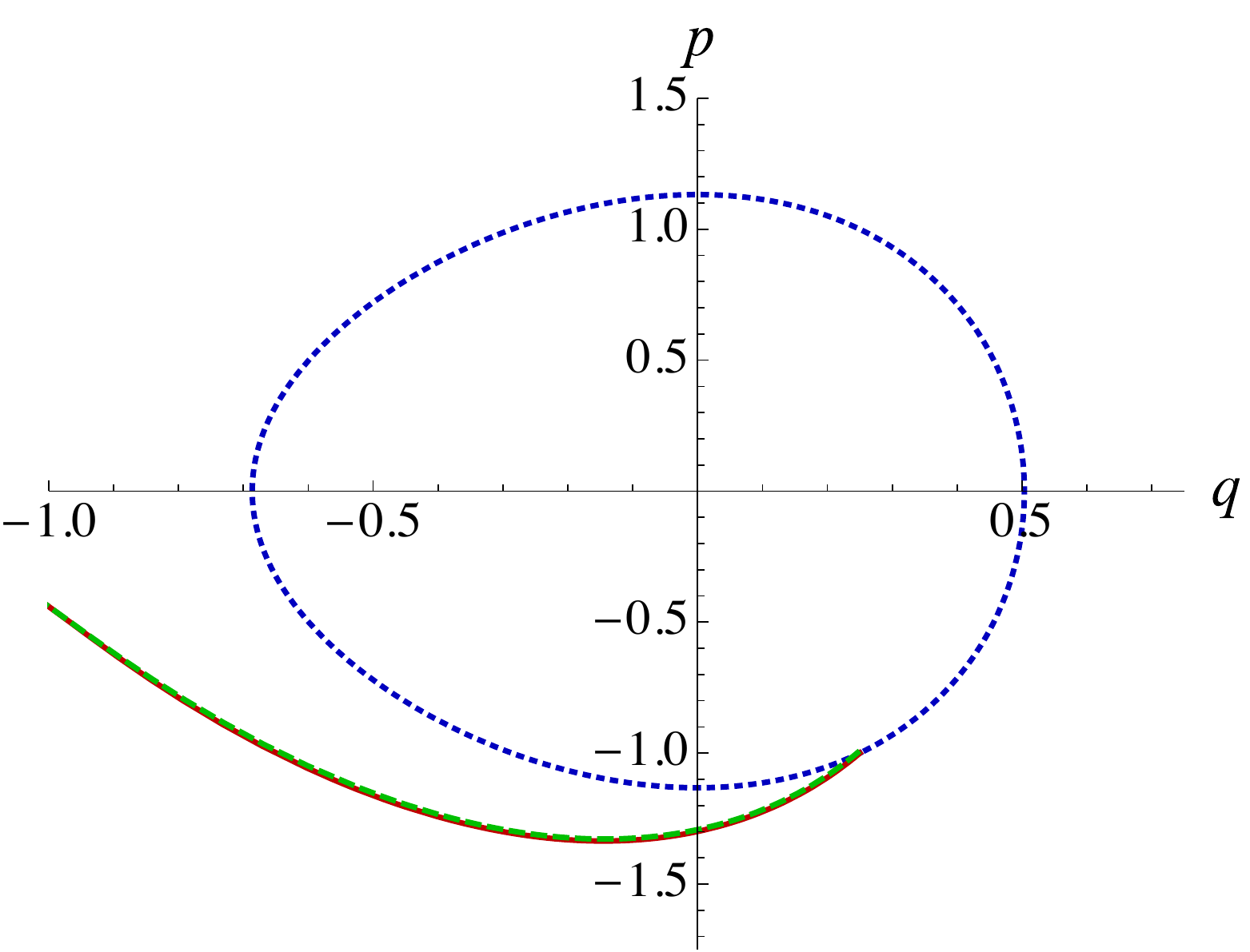}
  }
  \quad
  \subfigure[$n = 10$]{
    \includegraphics[width=.465\linewidth]{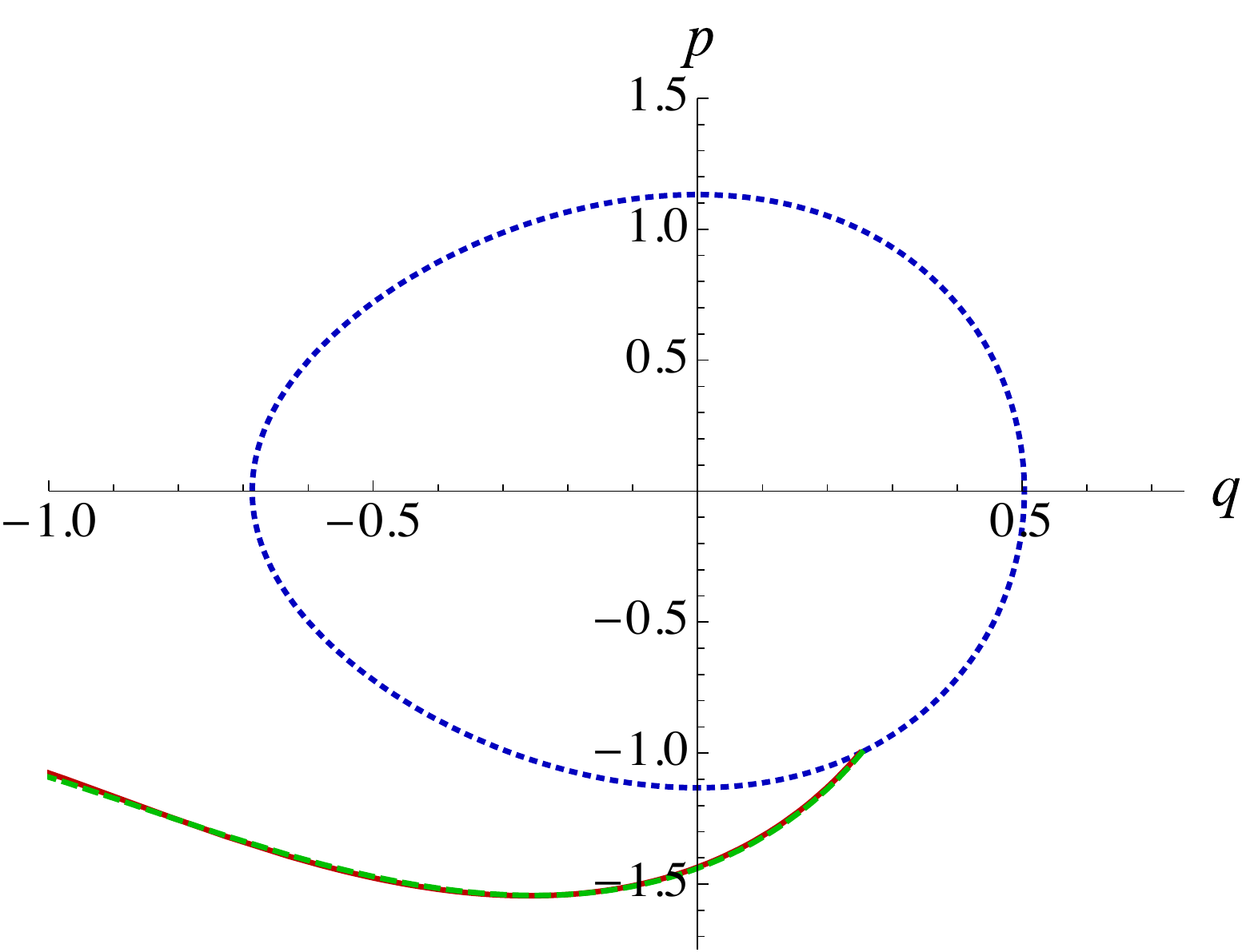}
  }
  \captionsetup{width=0.9\textwidth}
  \caption{
    Phase space plot of escape from cubic potential well; $\hbar = 0.05$, $n = 1, 3, 5, 10$. Plotted for $0 \le t \le T$ (unless the solution goes outside the range), where $T \simeq 3.39$ is the period of the classical solution.
    The solution to the classical Hamiltonian system (or the first two equations of \eqref{eq:Heller}; dotted and blue) is trapped inside the potential well.
    The solutions of the semiclassical system~\eqref{eq:SemiclassicalSystem} (solid and red) escape from the potential well inside, and show a very good agreement with the solutions (dashed and green) obtained by an algorithm based on Egorov's Theorem (or the IVR method) for a short time.
  }
  \label{fig:0.05}
\end{figure}

\begin{figure}[htbp]
  \centering
  \subfigure[$n = 1$]{
    \includegraphics[width=.465\linewidth]{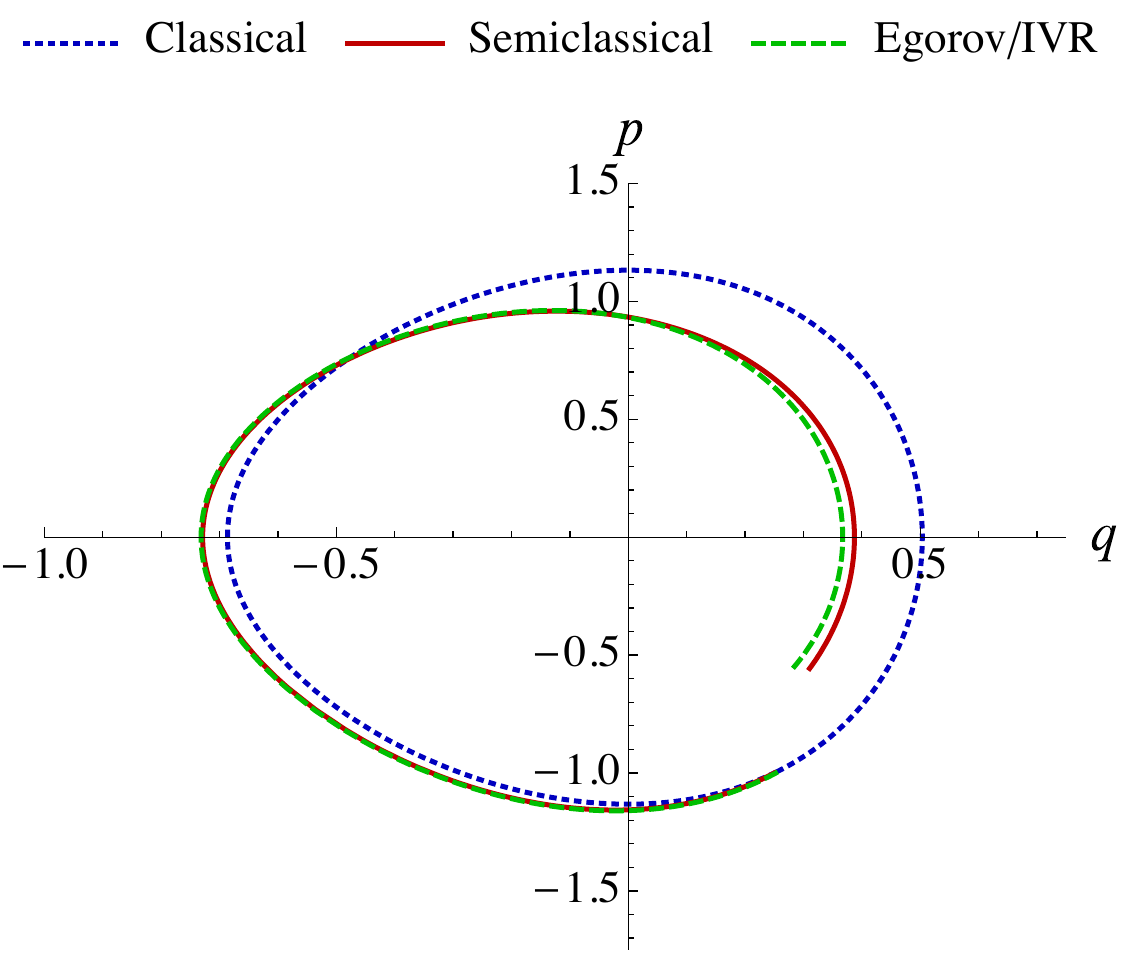}
  }
  \quad
  \subfigure[$n = 3$]{
    \includegraphics[width=.465\linewidth]{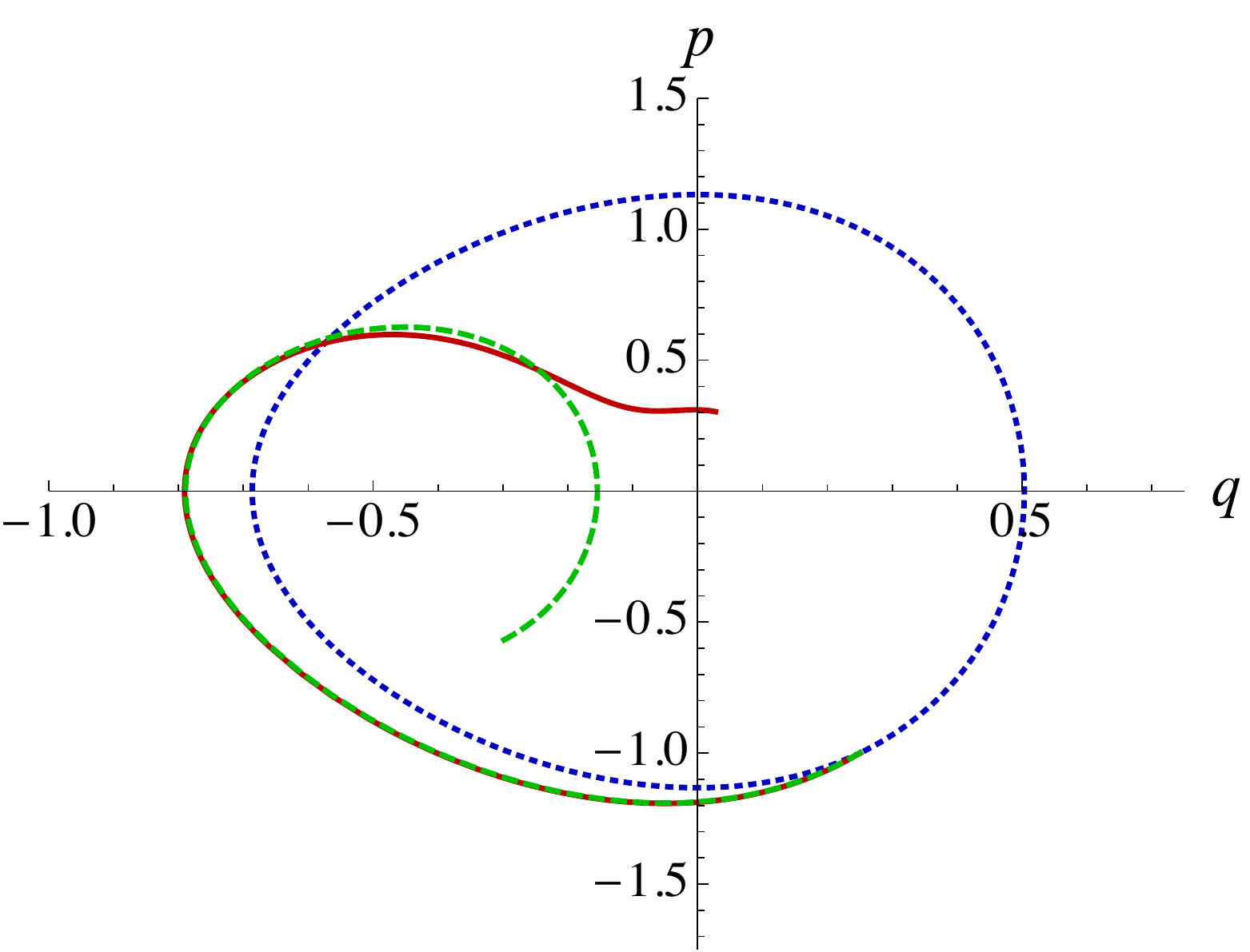}
  }
  \subfigure[$n = 5$]{
    \includegraphics[width=.465\linewidth]{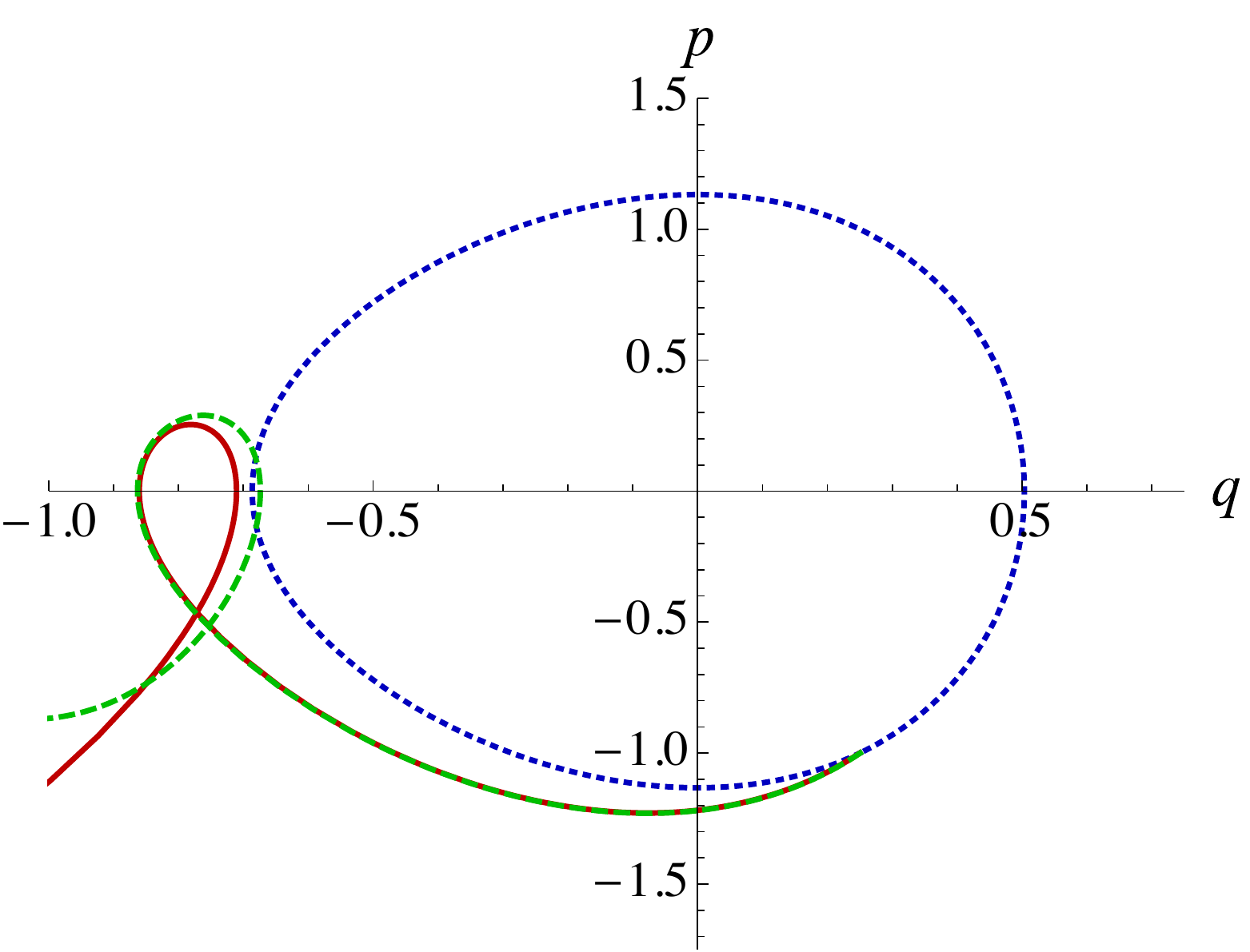}
  }
  \quad
  \subfigure[$n = 10$]{
    \includegraphics[width=.465\linewidth]{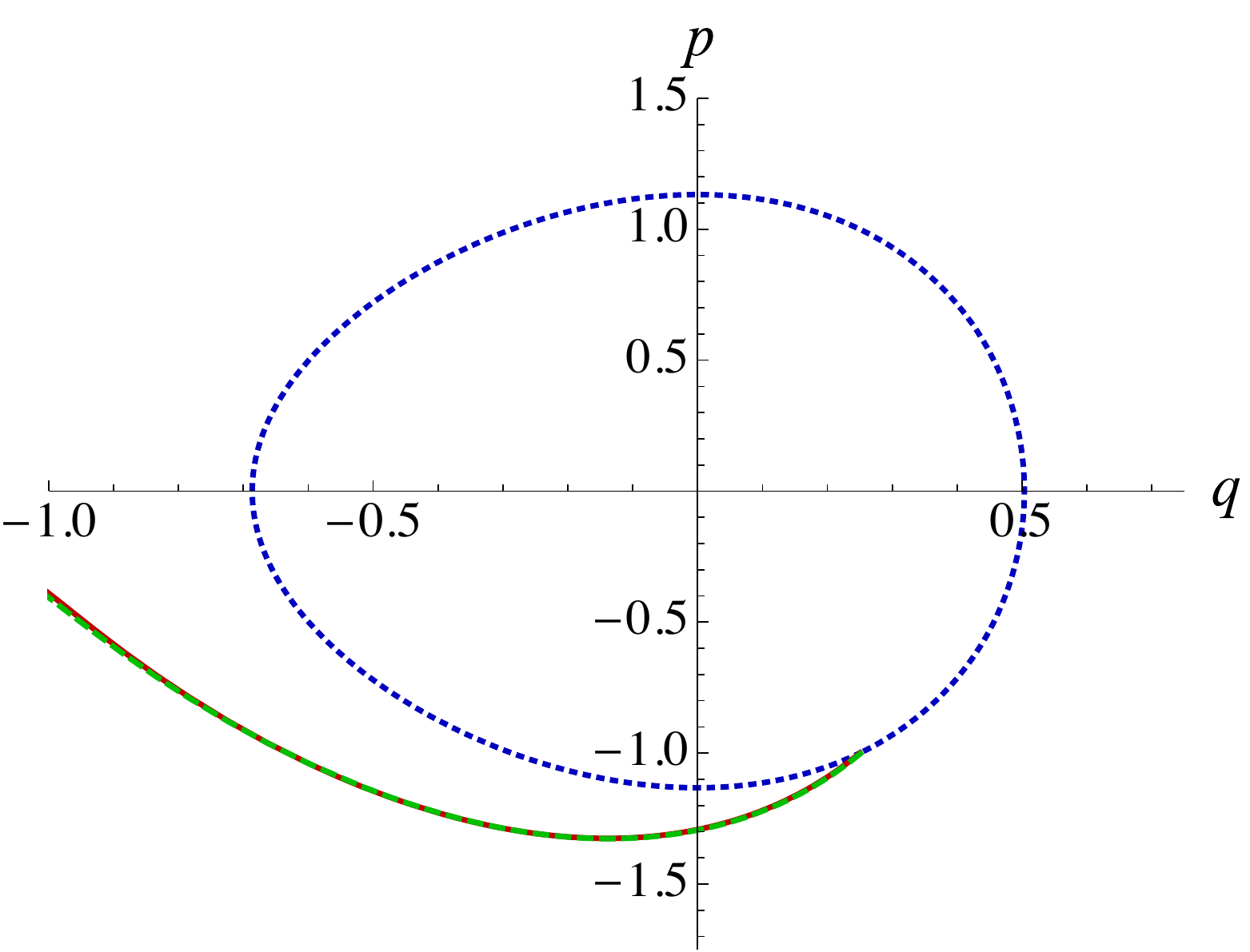}
  }
  \captionsetup{width=0.9\textwidth}
  \caption{Phase space plot of escape from cubic potential well: $\hbar = 0.025$}
  \label{fig:0.025}
\end{figure}

\begin{figure}[htbp]
  \centering
  \subfigure[$n = 1$]{
    \includegraphics[width=.465\linewidth]{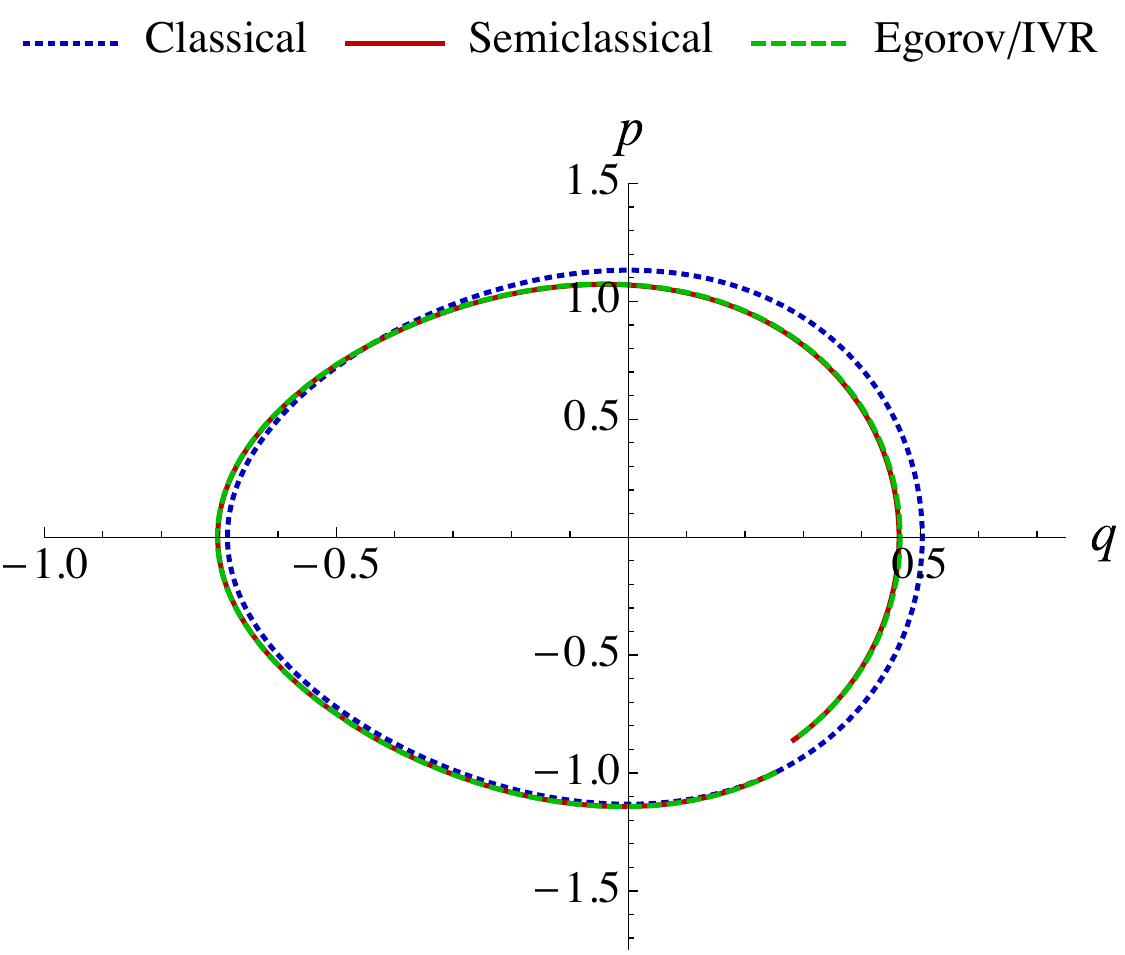}
  }
  \quad
  \subfigure[$n = 3$]{
    \includegraphics[width=.465\linewidth]{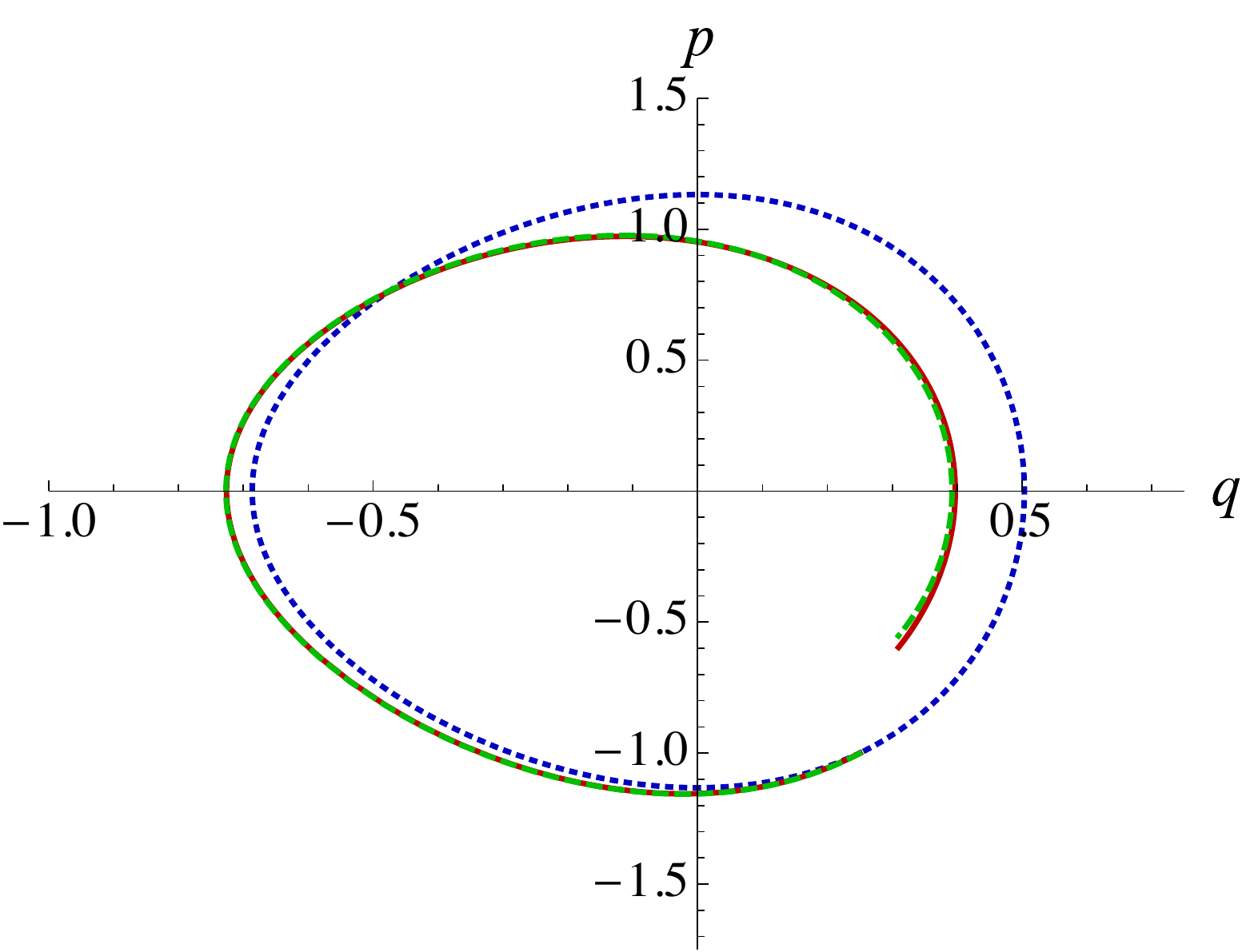}
  }
  \subfigure[$n = 5$]{
    \includegraphics[width=.465\linewidth]{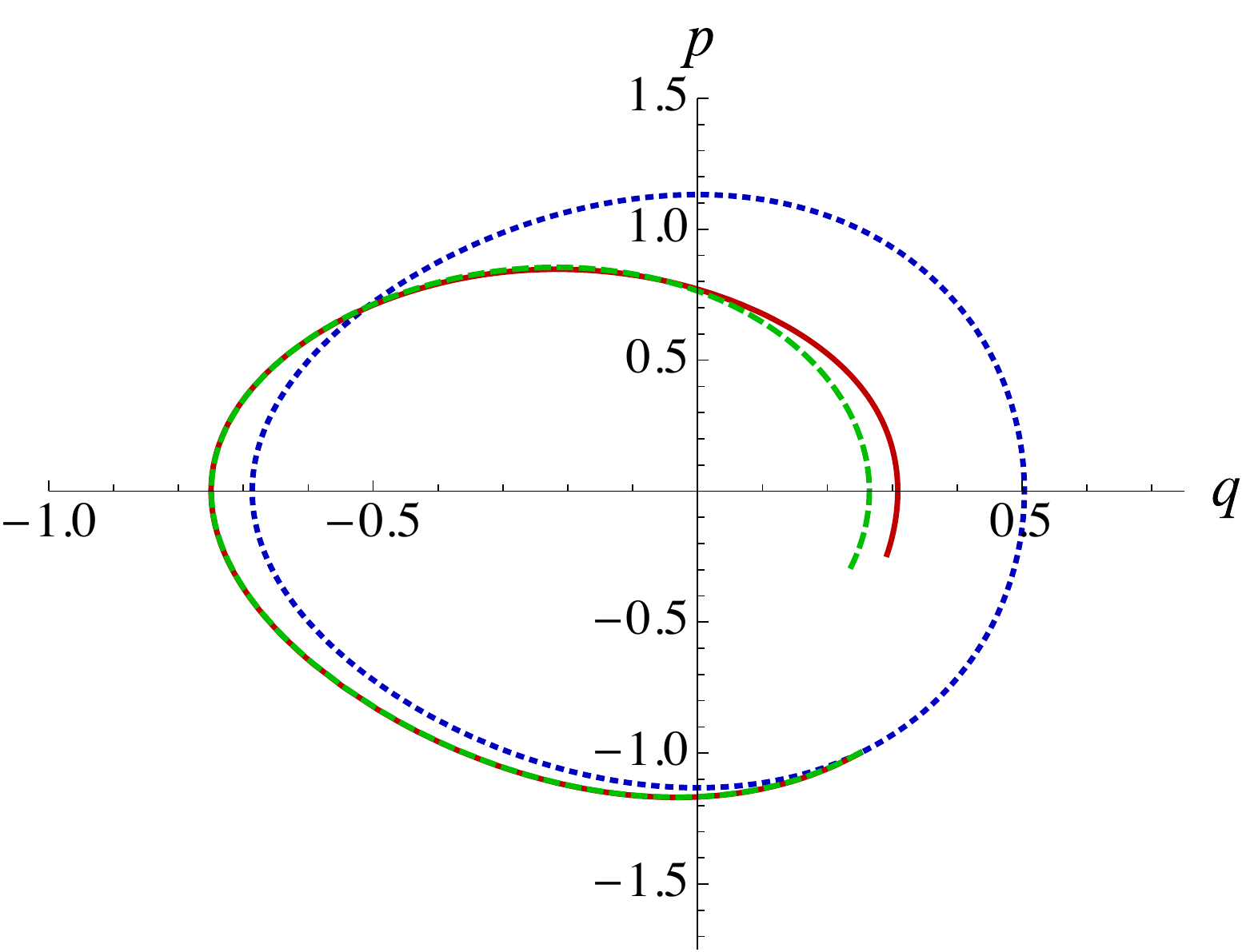}
  }
  \quad
  \subfigure[$n = 10$]{
    \includegraphics[width=.465\linewidth]{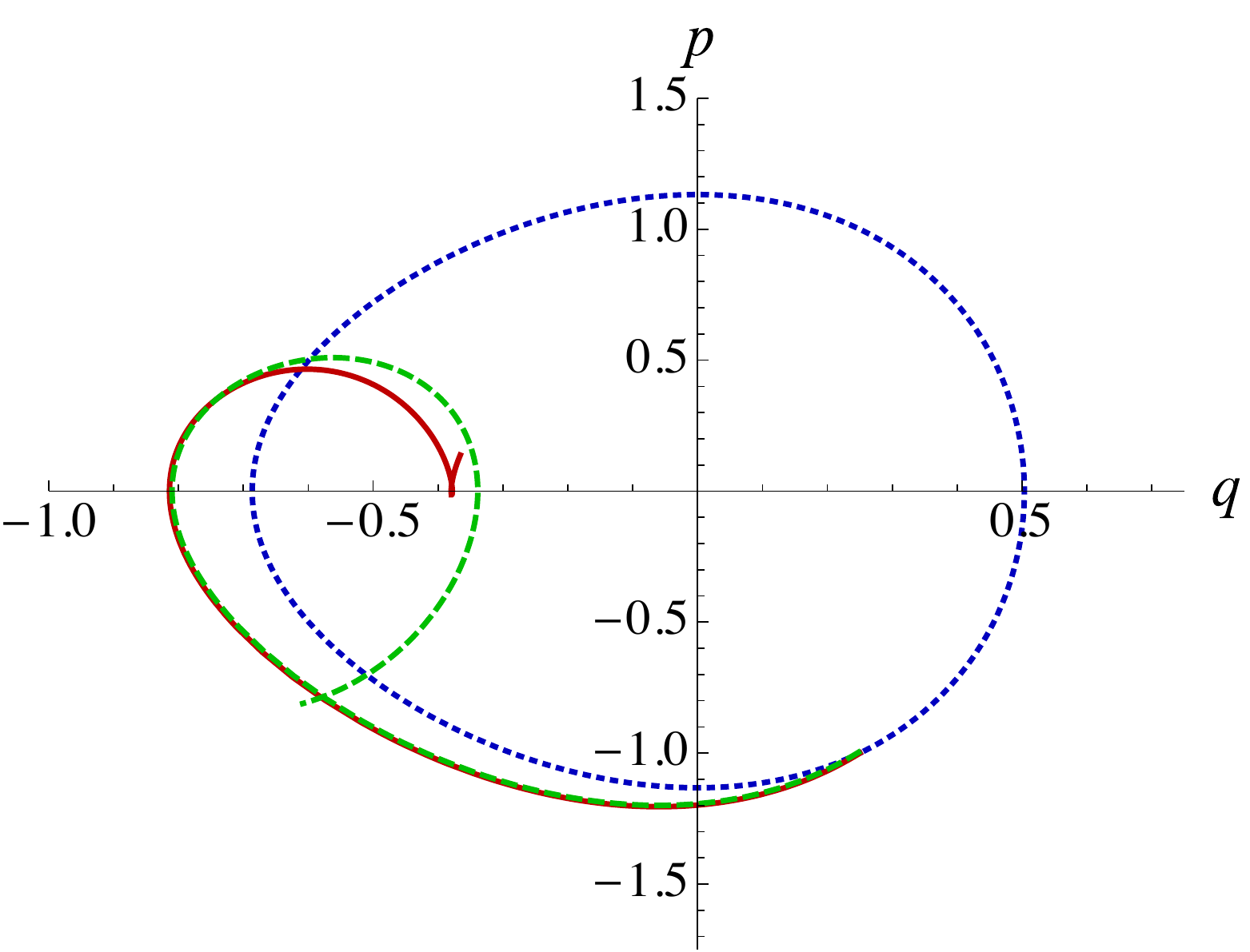}
  }
  \captionsetup{width=0.9\textwidth}
  \caption{Phase space plot of escape from cubic potential well: $\hbar = 0.01$}
  \label{fig:0.01}
\end{figure}

\begin{figure}[htbp]
  \centering
  \subfigure[$n = 1$]{
    \includegraphics[width=.465\linewidth]{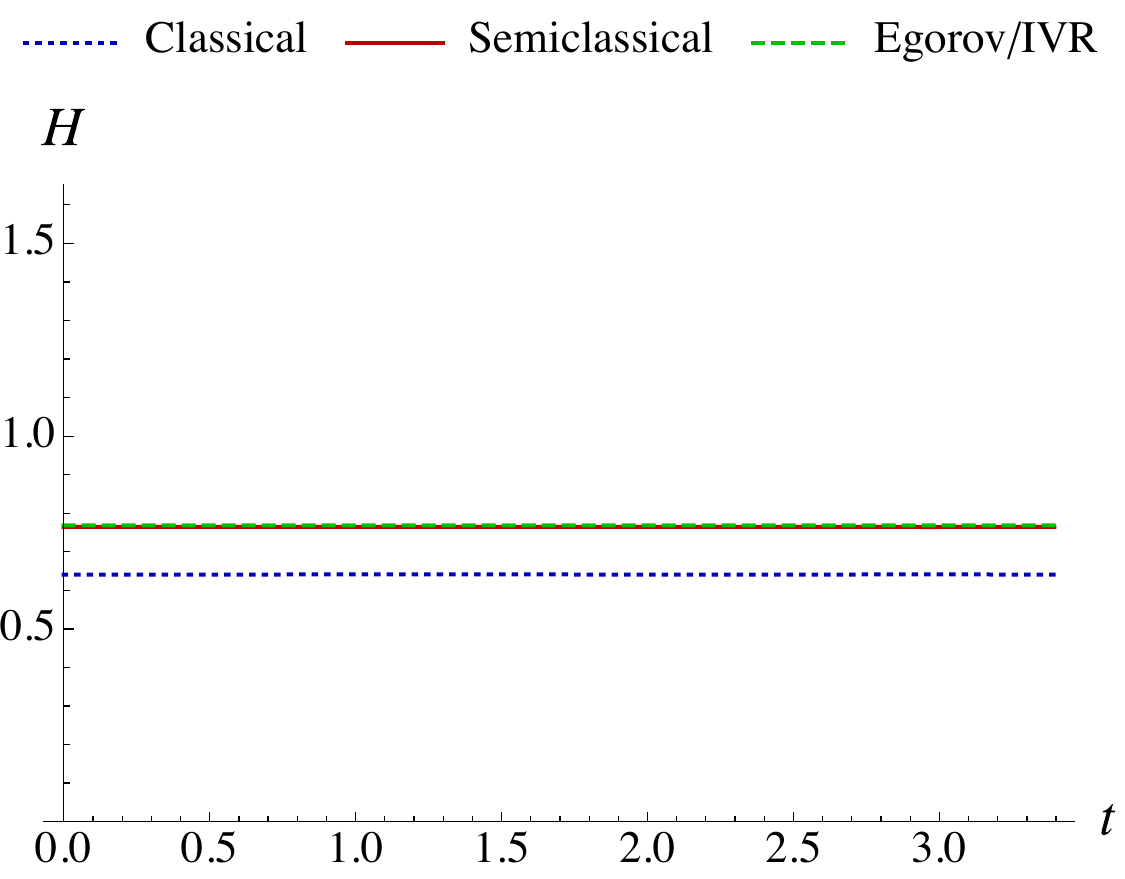}
  }
  \quad
  \subfigure[$n = 3$]{
    \includegraphics[width=.465\linewidth]{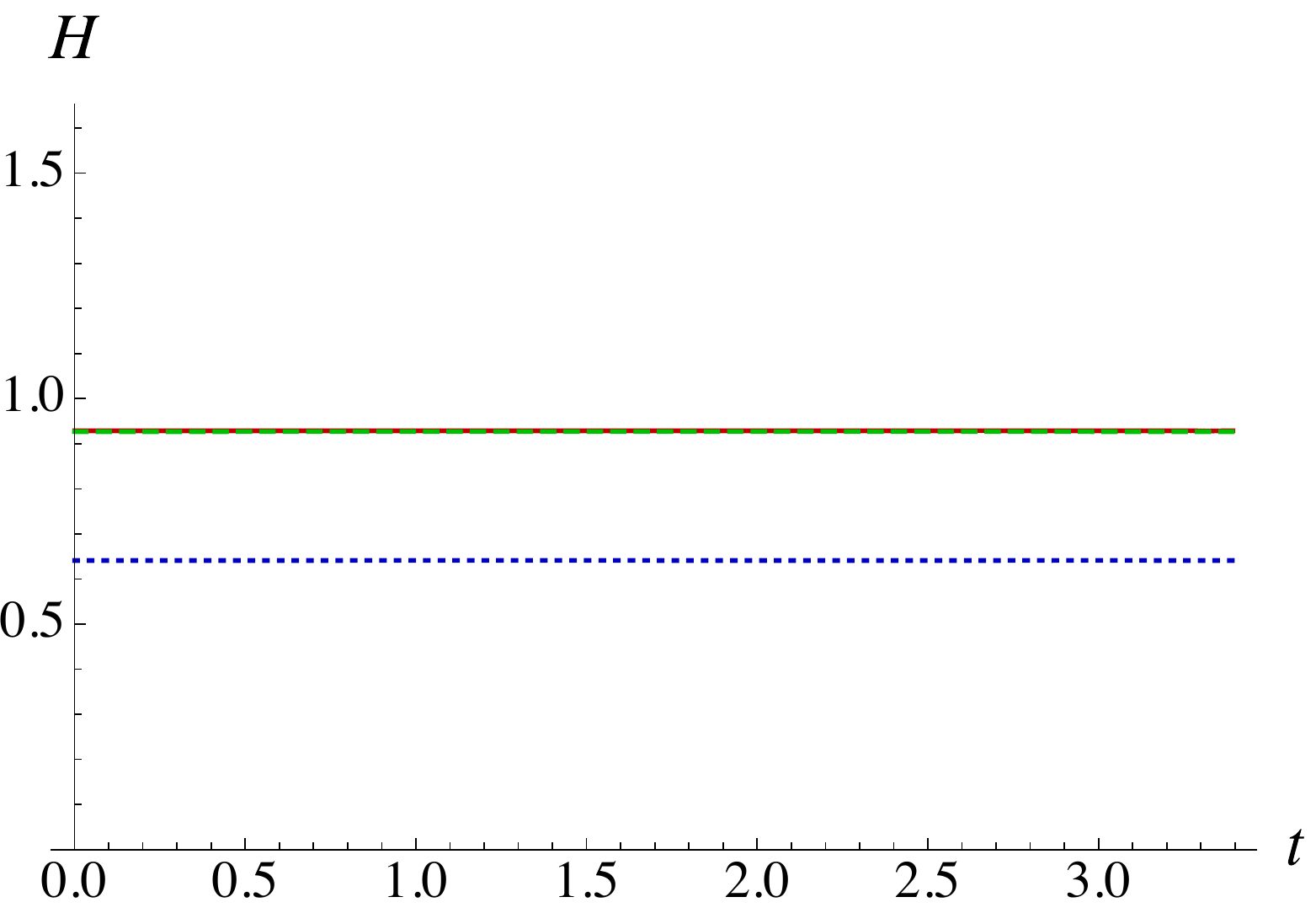}
  }
  \subfigure[$n = 5$]{
    \includegraphics[width=.465\linewidth]{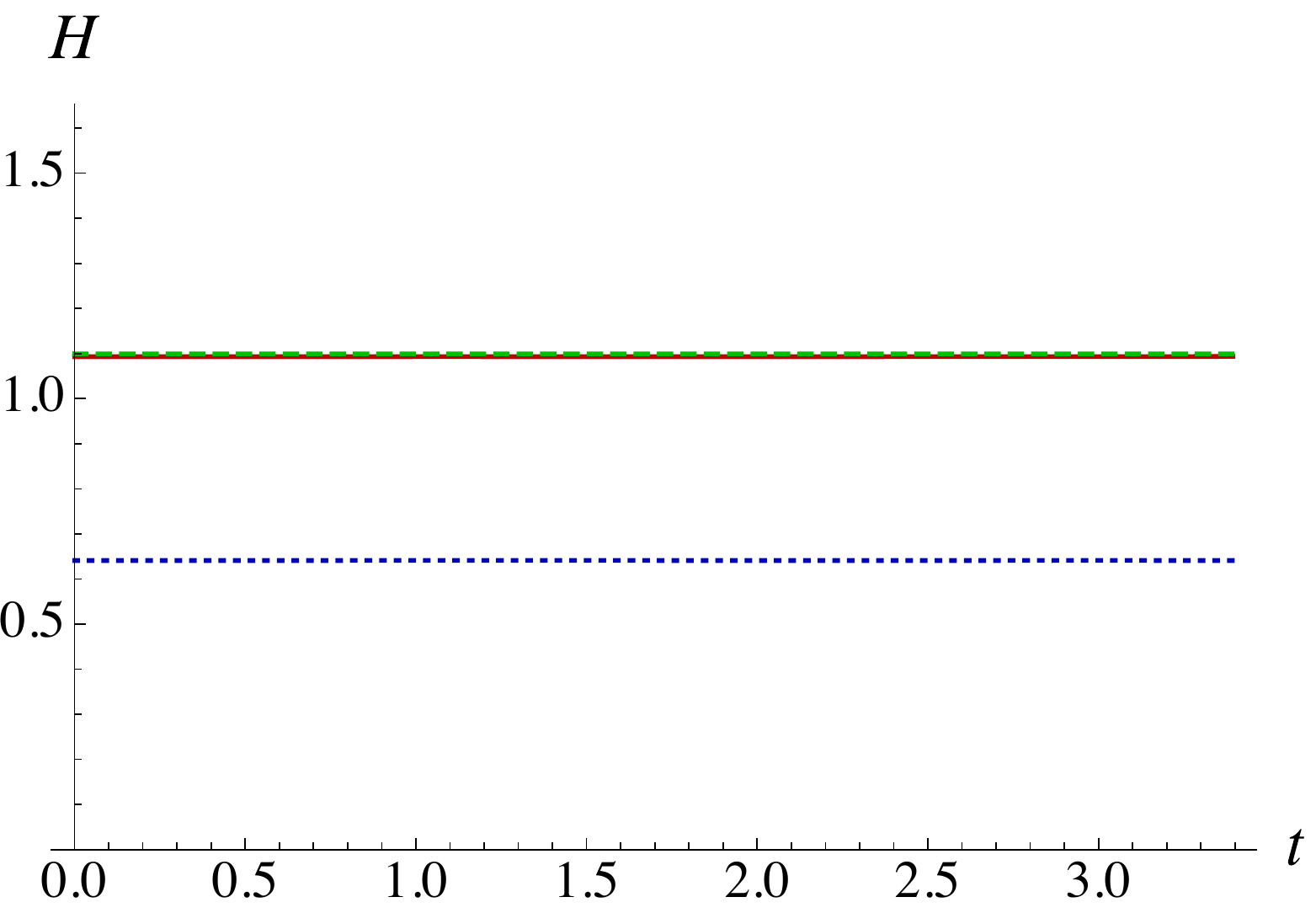}
  }
  \quad
  \subfigure[$n = 10$]{
    \includegraphics[width=.465\linewidth]{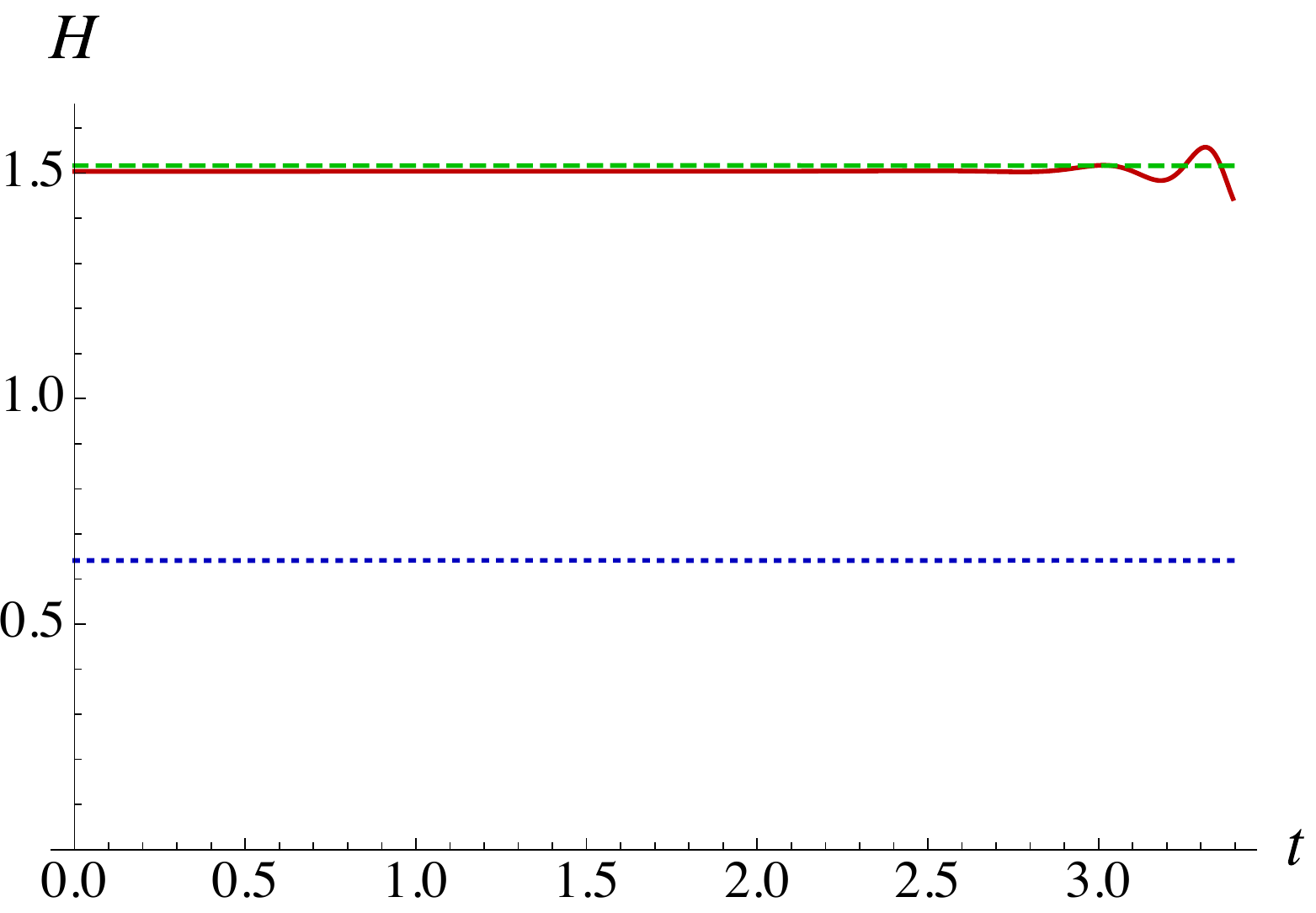}
  }
  \captionsetup{width=0.9\textwidth}
  \caption{Time evolution of the total energy or Hamiltonian for $\hbar = 0.025$.
    Classical energy $H_{\text{cl}} = p^{2}/(2m) + V$ along the classical solution, the semiclassical energy $\overline{H}^{(n)}_{\hbar}$ along the semiclassical solution of \eqref{eq:ReducedSemiclassicalSystem}, and the expectation value $\texval{\hat{H}}$ along the solutions of the Egorov-type algorithm~\cite{KeLaOh2016}.
    One can see that the energies are conserved well by the numerical methods.
    The oscillation of the semiclassical energy observed for $n = 10$ seems to be caused by the increased errors as $\mathcal{B}$ becomes very small, i.e., the wave packet becomes very widespread: $\mathcal{B}(t) \sim 0.001$ towards the end; note that the Ehrenfest time scale is $t \sim 3.7$ here.
  }
  \label{fig:0.025_t-H}
\end{figure}

\bibliography{SympSWPDyn}
\bibliographystyle{plainnat}

\end{document}